\documentclass{llncs}
\PassOptionsToPackage{usenames,dvipsnames}{xcolor}
\usepackage{wrapfig}

\usepackage{amssymb}
\usepackage{amsfonts}
\usepackage{amsmath} 
\usepackage{tikz}
\usepackage{pgf}
\usepackage{pgflibraryarrows}
\usepackage{pgffor}
\usepackage{pgflibrarysnakes}
\newcommand\todo[1]{}
\usepackage[usenames,dvipsnames]{xcolor}

\usetikzlibrary{fit} 
\usetikzlibrary{backgrounds} 
\usetikzlibrary{patterns}
\usetikzlibrary{positioning}
\tikzstyle{background}=[rectangle,fill=gray!10, inner sep=0.1cm, rounded corners=0mm]
\usepgflibrary{shapes}
\usetikzlibrary{snakes,automata}
\usetikzlibrary{shadows}

\tikzstyle{background}=[rectangle,fill=gray!10, inner sep=0.1cm, rounded corners=0mm]
\tikzstyle{loc}=[draw,rectangle,minimum size=1.4em,inner sep=0em]
\tikzstyle{trans}=[-latex, rounded corners]
\tikzstyle{trans2}=[-latex, dashed, rounded corners]

\newcommand{\PSPACEC}{\textsc{Pspace-Complete}}
\newcommand{\PSPACE}{\textsc{Pspace}}

\newcommand{\NPC}{\textsc{NP-Complete}}
\newcommand{\NLC}{\textsc{NL-Complete}}
\newcommand{\NLOGS}{\textsc{Nlogspace}}
\newcommand{\UNDEC}{\textsc{Undecidable}}

\newcommand{\PTIMEC}{\textsc{Ptime-Complete}}

\newcommand{\set}[1]{\left\{ #1 \right\}}

\newcommand{\seq}[1]{\langle #1 \rangle}
\newcommand{\Nat}{\mathbb N}

\newcommand{\R}{\mathbb R}
\newcommand{\Real}{\R}
\newcommand{\Int}{\mathbb Z}
\newcommand{\Rat}{\mathbb Q}
\newcommand{\Rplus}{\R_{\geq 0}}

\newcommand{\E}{{\bf E}}
\newcommand{\A}{{\bf A}}
\newcommand{\norm}[1]{\|#1\|}

\newcommand{\point}[1]{{\overline{#1}}}

\newcommand{\pt}{\point{t}}
\newcommand{\px}{\point{x}}
\newcommand{\py}{\point{y}}

\newcommand{\vv}{\vec{v}}

\newcommand{\vb}{\vec{b}}

\newcommand{\vr}{\vec{r}}

\newcommand{\vzero}{\vec{0}}

\newcommand{\sem}[1]{ [ \! [ {#1}  ]  \! ]} 

\newcommand{\Aa}{\mathcal{A}}

\newcommand{\Hh}{\mathcal{H}}
\newcommand{\Pp}{\mathcal{P}}
\newcommand{\Gg}{\mathcal{G}}

\newcommand{\N}{\mathbb{N}}
\newcommand{\Rr}{\mathcal{R}}

\newcommand{\Tt}{\mathcal{T}}

\newcommand{\poly}{\textrm{poly}}

\newcommand{\bB}{\mathbb{B}}

\newcommand{\Trace}{\mathsf{\it Trace}}

\newcommand{\until}{\:\mathcal{U}}

\usepackage[colorlinks=true]{hyperref}

\begin{document}

\title{Weak Singular Hybrid Automata~\thanks{This work was partly supported by
    IRCC project Spons/CS/I12155-1/2013.}}

\author{ 
	Shankara Narayanan Krishna
	\and Umang Mathur
	\and Ashutosh Trivedi
  }
\institute{Department of Computer Science and Engineering\\
	Indian Institute of Technology - Bombay\\
	Mumbai 400076, India
}

\maketitle

\begin{abstract}
  The framework of Hybrid automata---introduced by Alur, Courcourbetis,
  Henzinger, and Ho---provides a formal modeling and analysis environment to
  analyze the interaction between the discrete and the continuous parts of
  hybrid systems. 
  Hybrid automata can be considered as generalizations of finite state automata
  augmented with a finite set of real-valued variables whose dynamics in each
  state is governed by a system of ordinary differential equations. 
  Moreover, the discrete transitions of hybrid automata are guarded by
  constraints over the values of these real-valued variables, and enable
  discontinuous jumps in the evolution of these variables.
  Singular hybrid automata are a subclass of hybrid automata where dynamics is
  specified by state-dependent constant vectors.
  Henzinger, Kopke, Puri, and Varaiya showed that for even very restricted
  subclasses of singular hybrid automata, the fundamental verification questions,
  like reachability and schedulability, are undecidable. 
  Recently, Alur, Wojtczak, and Trivedi studied an interesting class of hybrid
  systems, called constant-rate multi-mode systems, where schedulability and
  reachability analysis can be performed in polynomial time. 
  Inspired by the definition of constant-rate multi-mode systems, in this paper
  we introduce   \emph{weak singular hybrid automata} (WSHA), a previously
  unexplored subclass of singular hybrid automata, and show the decidability
  (and the exact complexity) of various verification questions for this class
  including reachability (\NPC{}) and LTL model-checking (\PSPACEC{}).
  We further show that extending WSHA with a single unrestricted clock or
  with unrestricted variable updates lead to undecidability of
  reachability problem. 
\end{abstract} 

\section{Introduction}
\label{sec:introduction}
Hybrid automata, introduced by Alur et al.~\cite{ACHH92}, provide an intuitive
and semantically unambiguous way to model hybrid systems. Various
verification questions for such systems can then be naturally reduced to
corresponding questions for hybrid automata. 
Hybrid automata can be considered as finite state-transition
graphs with a finite set of real-valued variables with state-dependent dynamics
specified using a set of first-order ordinary differential equations. 
The variables of hybrid automata can be used to constrain the evolution of the 
system by means of \emph{guards} of the transitions and \emph{local invariants}
of the states of the state-transition graph.
The variables can also be reset at the time of taking a transition and thus
allowing discrete jumps in the evolution of the system.  
Considering the richness of the dynamics of hybrid automata, it should come as
no surprise that key verification questions, like state reachability,  are
undecidable for hybrid automata limiting the applicability of hybrid automata
for automatic verification of hybrid systems.
Henzinger et al.~\cite{HKPV98,HK99} observed that this negative result stays
even for a severely restricted subclass of hybrid automata, called the \emph{singular
hybrid automata} (SHA), where the variables dynamics is specified as state-dependent
constant-rate vectors and showed that the reachability problem stays
undecidable for singular hybrid automata with three clocks (unit-rate variables)
and one non-clock variable.   
In this paper we introduce a weak version of singular hybrid automata, and show
the decidability (and the exact complexity) of reachability, schedulability, and
LTL model-checking problems for this class.

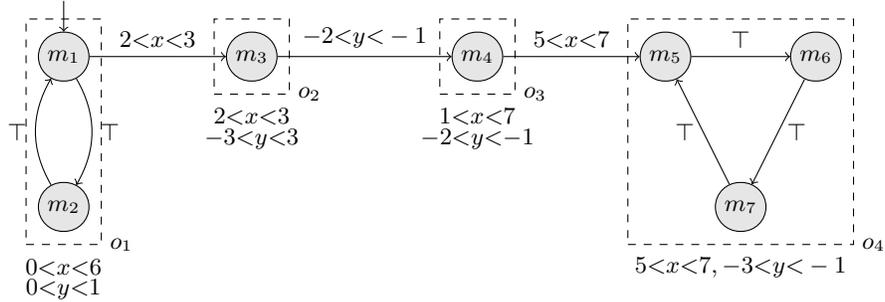
\begin{figure}[t]
  \begin{center}
    \begin{tikzpicture}
      \tikzstyle{every state}=[fill=gray!20!white,minimum size=2em,shape=rounded rectangle]

      \node[initial, initial where=above, initial text={},state,fill=gray!20] (m1) {$m_1$};
      \node[state,fill=gray!20] at (0, -2) (m2) {$m_2$};
      \draw [dashed] (-0.5, 0.5) rectangle (0.5, -2.5) node[right] {$o_1$};
      \node at (0, -2.8) {$0{<} x {<} 6$};
      \node at (0, -3.1) {$0 {<} y {<} 1$};

      \node[state, fill=gray!20] at (2.5,0) (m3) {$m_3$} ;
      \draw [dashed] (2, 0.5) rectangle (3, -0.5) node[right]{$o_2$} ;
      \node at (2.5, -0.8) {$2{<} x {<} 3$};
      \node at (2.5, -1.1) {$-3 {<} y {<} 3$};

      \node[state, fill=gray!20] at (5.5,0) (m4) {$m_4$} ;
      \draw [dashed] (5, 0.5) rectangle (6, -0.5) node[right]{$o_3$};
      \node at (5.5, -0.8) {$1{<} x {<} 7$};
      \node at(5.5, -1.1) {$-2 {<} y {<} {-}1$};

      \node[state, fill=gray!20] at (8,0) (m5) {$m_5$} ;
      \node[state, fill=gray!20] at (10,0) (m6) {$m_6$} ;
      \node[state, fill=gray!20] at (9,-2) (m7) {$m_7$} ;
      \draw [dashed] (7.5, 0.5) rectangle (10.5, -2.5) node[right]{$o_4$};
      \node at (9, -2.8) {$5 {<} x {<} 7, -3 {<} y {<} -1$};

     \path[->] (m1) edge[bend left] node [right] {$\top$}  (m2);
     \path[->] (m2) edge[bend left] node [left] {$\top$}  (m1);

     \path[->] (m1) edge node [above] {$2 {<} x {<} 3$}  (m3);
     \path[->] (m3) edge node [above] {$-2 {<} y {<} -1$}  (m4);

     \path[->] (m4) edge node [above] {$5 {<} x {<} 7$}  (m5);

    \path[->] (m5) edge node [above] {$\top$}  (m6);
    \path[->] (m6) edge node [right] {$\top$}  (m7);
    \path[->] (m7) edge node [left] {$\top$}  (m5);
    \end{tikzpicture}
  \end{center}
\caption{A weak singular hybrid automaton}
\label{fig:wsha}
\end{figure}
Our definition of weak singular hybrid automata is inspired by the
definition of constant-rate multi-mode systems (CMS)~\cite{ATW12}, that are
hybrid systems that can switch freely between a finite set of modes (or states)
and whose dynamics are specified by a finite set of variables with mode-dependent
constant rates.
The schedulability problem for CMS is to decide---for a given initial state and
convex and bounded safety set---whether there exists a non-Zeno mode-switching
schedule such that the system stays within the safety set.
On the other hand, the reachability problem is to decide whether there is a
schedule that steers the system from a given initial configuration to the target configuration
while staying within a specified bounded and convex safety set.  
Since the system is allowed to switch freely between the enabled modes, the
reachability and schedulability problems can be solved in polynomial
time~\cite{ATW12} by reducing them to a linear program.   
We say that a singular hybrid automaton is \emph{weak} if there exists an
ordering among the states such that the transition to a lower order state is
disallowed, and the states with the same ordering form a CMS, i.e. such states have a
common invariant and vacuous guards on transitions among themselves.

\begin{figure}
\begin{center}
  \begin{tikzpicture}
    \draw [fill=gray!20, opacity=0.8] (0,0) rectangle (6,1);
    \node at (3.5, 0.5) {$o_1$};
    \node[fill=white!80, circle,inner sep=0.2em, draw] at (0.5, 0.5) {${\color{black}s_0}$};
    \draw[->] (4.2, 0.2) -- (5.2, 0.8) node[right] {$m_1$};
    \draw[->] (4.2, 0.2) -- (5.2, 0.2) node[right] {$m_2$};
    
    \draw [fill=gray!30, opacity=0.8] (2,2) rectangle (3,-3);
    \node at (2.5, -2.8) {$o_2$};
    \draw[->] (2.5, 1.8) -- (2.5, 0.8) node[below] {$m_3$};
    
    \draw [fill=gray!40, opacity=0.8] (1, -1) rectangle (7, -2);
    \node at (3.8, -1.5) {$o_3$};
    \draw[->] (1.2, -1.5) -- (2.2, -1.5) node[right] {$m_4$};
    
    \draw [fill=gray!50, opacity=0.8] (5, -1) rectangle (7, -3);
    \node at (6.8, -2) {$o_4$};
    \draw[->] (6, -2) -- (6.5, -2.6) node[below] {$m_5$};
    \draw[->] (6, -2) -- (5.5, -2.6) node[below] {$m_6$};
    \draw[->] (6, -2) -- (6, -1.5) node[above] {$m_7$};
    \node[fill=white,circle,inner sep=0.2em, draw] at (5.4, -1.8)
    {${\color{black}s_T}$};

  \end{tikzpicture}                        
\end{center}
  \caption{Multi-mode system corresponding to a robotic motion planning problem}
  \label{fig:robocop}
\end{figure}
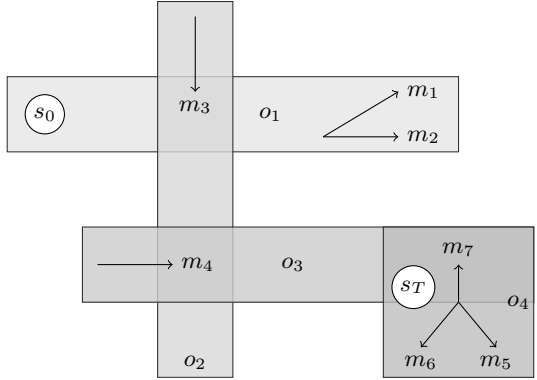
WSHAs are a natural generalization of CMS with structure, and can be used to model
CMS with non-convex safety set.  
As an example of a WSHA, consider the two dimensional robotic motion planning
problem shown in the Figure~\ref{fig:robocop}, where the arena is a nonconvex region
given as union of four convex polytopes $o_1, o_2, o_3$ and $o_4$. 
The possible motion primitives, or modes, in each region are shown as vectors 
showing the direction the robot will move given the corresponding mode is chosen. 
Consider the following reachability and schedulability problems for this
example: given an initial valuation $s_0$ decide if it is possible to compose the motion
primitives available in a given valuation so as to reach the final state $s_T$,
while the schedulability problem is to decide if there is a non-Zeno composition
of motion specifications such that the robot stays in the safety set forever.   
This problem can not be solved using the results for constant-rate multi-mode
systems sue to non-convexity of  the safety set.
On the other hand, it is easy to see that the reachability and the
schedulability problems for this system can be reduced to corresponding problems on
the weak singular hybrid automaton shown in Figure~\ref{fig:wsha}, where  modes
with the same order are shown inside a dashed box with global invariant  is
specified just below the box. 

We extend the results of~\cite{ATW12} by recovering decidability
for WSHA by showing that the reachability problem and
schedulability problems are \NPC{} for this model. 
We also define LTL and CTL model-checking problems for weak singular hybrid
automata, and show that while the complexity of LTL model-checking stays the
same as LTL model checking for finite state-transition graphs (\PSPACEC{}), the
CTL model-checking is already PSPACE-hard.
Inspired by an unpublished result from Bouyer and Markey~\cite{markey-HDR11}, we
show (Section~\ref{sec:undec}) that extending WSHA with single unrestricted
clock variable make the reachability problem undecidable for WSHA with three
variables.  
In the same section, we also show that extending WSHA with unrestricted variable
updates also make the reachability problem undecidable for WSHA with three
variables.  
Table~\ref{table:results} shows a summary of results on singular hybrid
automata, and the contributions of this paper are highlighted with boldface.  
\begin{center}
\begin{table}[h]
  \begin{tabular}{ p {3cm}   p {4.8 cm}  p {4cm} }
    \hline
    Problem  & SHA & WSHA \\
    \hline        
    Reachability  & \UNDEC{} (${\geq} 3$ vars{.})~\cite{HKPV98} \newline
    \textbf{{\NLC{} ($1$ var.)}} &  \textbf{{\NPC}}\\
    \hline
    Schedulability  & \UNDEC{} (${\geq} 3$ vars{.})~\cite{HKPV98} \newline \textbf{{\NLC{} ($1$ var.)}} &
    \textbf{{\NPC{}}} \\
    \hline
    LTL model-checking & \UNDEC{} (${\geq} 3$ vars{.})~\cite{HKPV98} \newline  \textbf{{\PSPACEC{} ($1$ var{.})}} &
    \textbf{{\PSPACEC{}}}\\
    \hline
    CTL model-checking & \textbf{{\UNDEC{} (${\geq} 2$ vars{.} )}} \newline \textbf{{\PTIMEC{} ($1$ var{.})}} &
    \textbf{{\PSPACE-Hard{}  (${\geq} 2$ vars. )}} \newline
    \textbf{{\PTIMEC{} ($1$ var.)}} \\
    \hline
  \end{tabular}
\vspace{1em}
 \caption{Summary of decidability results related to (weak) singular
    hybrid automata}
  \label{table:results}
\vspace{-2em}
\end{table}
\end{center}

\noindent{\bf Related work}.  
 Timed automata are subclasses of SHA with the restriction that all variables
 are clocks, while stopwatch automata are subclasses of hybrid
 automata with the  restriction  that all variables are stopwatches (clocks that
 can be paused).
 Using the region construction~\cite{AD94} Alur and Dill showed that the
 reachability and the schedulability problems for timed automata are 
 decidable and are in fact complete for PSPACE.
 Henzinger et al.~\cite{HKPV98} showed that the undecidability result for
 singular hybrid automata holds even for stopwatch automata. 
 Initialized singular hybrid automata are subclasses of singular hybrid
 automata with the restriction that if there is a transition between two
 modes that have different rate for some variable then that transition must
 reset that variable.
 Henzinger at al.~\cite{HKPV98} showed the decidability of reachability problem
 by reducing the problem to the corresponding problem on timed
 automata---by appropriate adjustment of the guards of the transitions. 
 Unlike timed automata and initialized SHA, our results for WSHA do not rely on the
 existence of finitary  bisimulation.   

Asarin, Maler, and Pnueli~\cite{AMP95} studied a subclass of singular hybrid
automata, called the piecewise-constant derivative (PCD) systems, that are
defined by a partition of the Euclidean space into a finite set of
polyhedral regions, where the dynamics in a region is defined by a constant rate
vector.   
PCD systems, unlike our model,  are defined as completely deterministic systems
where discrete transitions occurs at region boundaries and runs
change their directions according to the rate vector available in the new region.
They showed that even under such simple dynamics the reachability problem for PCD
systems with three or more variables is undecidable~\cite{AMP95}. 
The work that is closest to the results in this paper is on constant-rate
multi-mode systems by Alur et al.~\cite{ATW12,AFMT13}. 
However, our model strictly generalizes this model and permits analysis of
multi-mode systems with non-convex safety set. 
On the positive side, Asarin, Maler, and Pnueli~\cite{AMP95} gave an algorithm
to solve the reachability problem for two-dimensional PCD systems.

The paper is organized as the following.
In the next section we introduce technical notations and background required for
 the paper.
In Section~\ref{sec:reach-sched} we present weak singular hybrid automata and
show the decidability and complexity results for the reachability,
schedulability, and LTL model-checking problems.  
In Section~\ref{sec:undec} we present the two undecidability results related to
WSHA.
Due to lack of space the details of most of the proofs are in appendix.

\section{Preliminaries}
\label{sec:definitions}
 Let $\Real$ be the set of real numbers.
 Let $X$ be a finite set of real-valued variables.
 A \emph{valuation} on $X$ is a function $\nu : X \to \Real$.
 We assume an arbitrary but fixed ordering on the variables and write $x_i$
 for the variable with order $i$. 
 This allows us to treat a valuation $\nu$ as a point $(\nu(x_1), \nu(x_2),
 \ldots, \nu(x_n)) \in \Real^{|X|}$. 
  Abusing notations slightly, we use a valuation on $X$ and a point in
  $\Real^{|X|}$ interchangeably. 
  We denote points in this state space by $\px, \py$,  vectors by $\vr, \vv$, and 
  the $i$-th coordinate of point $\px$ and vector $\vr$ by $\px(i)$ and $\vr(i)$,
  respectively. 
  We write $\vzero$ for a vector with all its coordinates equal to $0$.
  We say that a set $S \subseteq \Real^n$ is {\em bounded} if there exists
  $d \in \Rplus$ such that for all $\px, \py \in S$ we have
  $\norm{\px - \py} \leq d$.

 We define a constraint over a set $X$ as a subset of $\Real^{|X|}$.
 We say that a constraint is \emph{polyhedral} if it is defined as the conjunction
 of a finite set of linear constraints of the form 
 $a_1 x_1 + \dots + a_n x_n \bowtie k,$
 where $k \in \Int$, for all $1 \leq i \leq n$ we have that 
 $a_i \in \Int, x_i \in X$, and $\bowtie \in \{<,\leq, =, >, \geq\}$.   
 Every polyhedral constraints can be written in the standard form $A\px \leq
 \vb$ for some matrix $A$ of size $k \times n$ and a vector $\vb \in \Int^k$.
 We call a bounded polyhedral constraint  a \emph{convex polytope}.
 For a constraint $G$, we write $\sem{G}$ for the set of valuations in
 $\Real^{|X|}$ satisfying the constraint $G$.  
 We write $\top$ ( resp., $\bot$) for the special constraint that is true
 (resp., false) in all the valuations, i.e. $\sem{\top} = \Real^{|X|}$ 
(resp., $\sem{\bot} = \emptyset$). 
 We write $\poly(X)$ for the set of polyhedral constraints over $X$ including
 $\top$ and $\bot$.  

\subsection{Singular Hybrid Automata}
Singular hybrid automata extend finite state-transition graphs with a finite set
of real-valued variables that grow with state-dependent constant-rates. 
The transitions of the automata are guarded by predicates on the valuations of
the variables, and the syntax allows discrete update of the value of the
variables. 
 
 \begin{definition}[Singular Hybrid Automata]
   A singular hybrid automaton is a tuple 
   $(M, M_0, \Sigma, X, \Delta, I, F)$ where:
  \begin{itemize}
  \item 
    $M$ is a finite set of control \emph{modes} including a distinguished
    initial set of control modes $M_0 \subseteq M$, 
  \item 
    $\Sigma$ is a finite set of \emph{actions},
  \item 
    $X$ is an (ordered) set of \emph{variables}, 
  \item
    $\Delta \subseteq M \times \poly(X) \times \Sigma \times 2^{X}
    \times M$ is the \emph{transition relation}, 
  \item 
    $I: M \to \poly(X)$  is the mode-invariant function, and
    \item 
    $F: M \to \Rat^{|X|}$ is the mode-dependent \emph{flow function}
    characterizing the rate of each variable in each mode.
  \end{itemize}
  For computation purposes, we assume that all real numbers are rational and
  represented by writing down the numerator and denominator
  in binary.
\end{definition}

For all $\delta = (m, G, a, R, m') \in \Delta$ we say that $\delta$ is a
\emph{transition} between the modes $m$ and $m'$ with \emph{guard} $G \in
\poly(X)$ and reset set $R \in 2^{X}$. 
For the sake of notational convenience and w.l.o.g., we assume that an action
$a \in \Sigma$ uniquely determines a transition $(m, G, a, R, m')$, and we write
$G(a)$ and $R(a)$ for the guard and the reset set corresponding to the action $a
\in \Sigma$.
This can be assumed without loss of generality, since, in this paper, we do not
study language-theoretic properties of an SHA, and assume that the
non-determinism is resolved by the controller. 

A \emph{configuration} of a SHA $\Hh$ is a pair $(m, \nu) \in M \times
\Real^{|X|}$ consisting of a control  mode $m$ and a variable valuation $\nu {\in}
\Real^{|X|}$ such that that $\nu$ satisfies the invariant $I(m)$ of the mode $m$,
i.e. $\nu \in \sem{I(m)}$.   
We say that the transition $\delta = (m, G, a, R, m')$ is \emph{enabled} in a configuration $(m,
\nu)$ when guard $G \in \poly(X)$ is satisfied by the valuation, i.e. $ \nu \in
\sem{G}$.
Moreover, the transition $\delta$ resets the variables in $R \in 2^{X}$ to $0$.   
We write $\nu[R{:=}0]$ to denote the valuation resulting from
substituting in valuation $\nu$ the value for the variables in the set $R$ to
$0$, formally  $\nu[R{:=}0](x) = 0$ if $x \in R$ and $\nu[R{:=}0](x) = \nu(x)$
otherwise.
A \emph{timed action} of a SHA is the tuple $(t, a) \in \Rplus \times \Sigma$
 consisting of a time delay and discrete action. 
 While the system dwells in a mode $m \in M$ the valuation of the system
 flows linearly according to the rate function $F(m)$, i.e.  after spending
 $t$ time units in mode $m$ from a valuation $\nu$ the valuation of the variables
 will be $\nu + t \cdot F(m)$.

 We say that $((m, \nu), (t, a), (m', \nu'))$ is a transition of a SHA $\Hh$
 and we write $(m, \nu) \xrightarrow{t}_{a} (m', \nu')$ 
 if $(m, \nu)$ and $(m', \nu')$ are valid configurations of the SHA $\Hh$, and
 there is a transition  $\delta = (m, G, a, R, m') \in \Delta$ such that: 
 \begin{itemize}
 \item  
   all the valuations resulting from dwelling in mode $m$ for time $t$ from the
   valuation $\nu$ satisfy the invariant of the mode $m$, i.e. 
   $(\nu + F(m) \cdot \tau) \in \sem{I(m)}$ for all $\tau \in [0, t]$ 
   (observe that due to convexity of the invariant set we only need to check
   that  $(\nu + F(m) \cdot t) \in  \sem{I(m)}$); 
 \item
   The valuation reached after waiting for $t$ time-units satisfy the
   constraint $G$ (called the  guard of the transition $\delta$), i.e. 
   $(\nu + F(m) \cdot t) \in \sem{G}$, and
 \item
   $\nu' = (\nu + F(m)\cdot t)[R:=0]$. 
 \end{itemize}

 A \emph{finite run} of a singular hybrid automaton $\Hh$ is a finite sequence 
 $r = \seq{(m_0, \nu_0), (t_1, a_1), (m_1, \nu_1), (t_2, a_2), \ldots, (m_k,
   \nu_k)}$  
 such that $m_0 \in M_0$ and for all $0 \leq i < k$ we have
 that $((m_i, \nu_i), (t_{i+1}, a_{i+1}), (m_{i+1}, \nu_{i+1}))$ is a transition
 of $\Hh$.
 For such a run $r$ we say that $\nu_0$ is the \emph{starting valuation}, while
 $\nu_k$ is the \emph{terminal valuation}. 
 An \emph{infinite run} of an SHA $\Hh$ is similarly defined to be an infinite
 sequence $r = \seq{(m_0, \nu_0), (t_1, a_1), (m_1, \nu_1), (t_2, a_2), \ldots}$ 
 such that $((m_i, \nu_i), (t_{i+1}, a_{i+1}), (m_{i+1}, \nu_{i+1}))$ is a
 transition of the SHA $\Hh$ for all $i \geq 0$.
 We say that $\nu_0$ is the starting configuration of the run. 
 We say that such an infinite run 
   is Zeno if $\sum_{i=1}^{\infty} t_i < \infty$. 
 Zeno runs are physically unrealizable since they require infinitely many
 mode-switches within a finite amount of time.

\subsection{Reachability, Schedulability, and Model-Checking}
Given a finite set of atomic propositions $\Pp$ and a labeling function $L: M
{\to} 2^\Pp$, a trace of a SHA $\Hh$
corresponding to an infinite run $r=\seq{(m_0, \nu_0), (t_1, a_1), \ldots}$
is the sequence $\seq{L(m_0), L(m_1), L(m_2),  \dots L(m_n), \ldots}$ of labels
corresponding to the mode sequence of $r$.
We use the standard syntax and semantics of LTL and CTL~\cite{BK08} with the
exception that we consider traces corresponding to non-Zeno runs. 
Given a SHA $\Hh = (M, M_0,$ $\Sigma, X, \Delta, I, F)$ and a starting valuation
$\nu \in \Real^{|X|}$, we are interested in the following problems over SHA.
\begin{itemize}
\item {\bf Reachability problem}.
  Given a target polytope $\Tt \subseteq \Real^{X}$, decide whether there  exists a
  finite run from $\nu_0$ to some valuation $\nu' \in \Tt$.  
\item {\bf Schedulability}. 
  Decide whether there exists an infinite non-Zeno run starting from $\nu$.
\item {\bf LTL model checking}.
  Given a set of propositions $\Pp$, labeling function $L$, and an
  LTL formula $\phi$ decide whether all non-Zeno traces of $\Hh$ satisfy
  $\phi$. 
\item {\bf CTL model checking}.
  Given a set of propositions $\Pp$, labeling function $L$, and a CTL formula
  $\phi$ decide whether all initial modes of $\Hh$ satisfy $\phi$.  
\end{itemize}

The termination~\cite{Min67} and the recurrent computation~\cite{AH94} problems
for two-counter Minsky machines are known to be undecidable.
By encoding the two counters as two vairables, and using another variable to
do additional book-keeping, the termination and the recurrence problem for
Minsky machines can be reduced to reachability and schedulability problems for
SHA. 
\begin{theorem}[Undecidability~\cite{HKPV98,AMP95,AM98}]
  \label{thm:SHA-undec-reach}
  The reachability, schedulability, LTL and CTL model-checking  problems are
  undecidable for SHA with three variables.  
\end{theorem}

\subsubsection{Improved complexity results.}

Using just two variables $x,y$ (of which $y$ is only a clock variable), with the encoding 
$x=2-\frac{1}{2^{c_1}3^{c_2}}$ for counters $c_1,c_2$, we improve the undecidability 
result for CTL model checking of SHAs:   
\begin{theorem}
  \label{thm:ctl-mc-undec}
  CTL Model-checking problem for singular hybrid automata with two variables is
  $\UNDEC$. 
\end{theorem}

We adapt the construction of Laroussinie et al.~\cite{LMS04}
for the case of one clock timed automata to show the following results for
SHA  with one variable. 
 \begin{theorem}
   \label{thm:dec-SHA-one-var}
   For SHA with one variable we have the following results.
   \begin{itemize}
   \item[(a)]
     The reachability and the schedulability problems are $\NLC$.
   \item[(b)]
     LTL Model-checking problem is $\PSPACEC{}$. 
   \end{itemize}
 \end{theorem}
 \begin{proof}{(Sketch.)}
  The $\NLOGS$-hardness of the reachability problem for SHA follows from the
  complexity of reachability problem for finite graphs~\cite{jones1976new}.  
  On the other hand, the $\NLOGS$-hardness of the schedulability problem for SHA follows
  from the complexity of nonemptiness problem for B\"uchi automata (Proposition
  10.12 of ~\cite{PP04}).  
  For $\NLOGS$-membership of these problems we adapt the region construction
  (LMS regions) for 
  one-clock timed automata  proposed by Laroussinie, Markey, and
  Schnoebelen~\cite{LMS04}.  
  
  For LTL model checking the \PSPACE-hardness follows from the
  PSPACE-completeness~\cite{SC85} for LTL model checking on finite automata,
  while the PSPACE membership follows from the region construction introduced in
  the proof of (a).   
  However, we need to be extra careful as our semantics are defined with respect
  to non-Zeno runs.  
  To overcome this complication, we characterize non-Zenoness property of the
  region graphs as LTL formulas using the following Lemma. 
%
\begin{lemma}
  \label{thm:sha-one-var-nz-main}
  Let $c_x$ and $C_x$ denote the smallest and largest constants 
  used in guards of $\Hh$ and let $\Rr\Gg_{\Hh}$ be the LMS region graph of $\Hh$.
  An infinite run in the region graph $\Rr\Gg_{\Hh}$  of the form 
  $((m_{0},r_{0}), a_{0}, (m_{1},r_{1}), a_{1}, \ldots)$ is called
  \emph{progressive} iff it has a non-Zeno instantiation.  
  Here, $r_{j}$ and $m_{j}$ are respectively the regions and 
  modes. 
  An infinite run $((m_{0},r_{0}), a_{0}, (m_{1},r_{1}), a_{1}, \ldots)$ 
  in the region graph $\Rr\Gg_{\Hh}$ of a one-variable SHA $\Hh$ is progressive
  iff one of the following hold: 
  \begin{enumerate}
  \item \label{nz-1} 
    For all $j{\geq} 0$ there exists $k>j$ such that $F(m_{k}) = 0$;  
  \item \label{nz-2} 
    There exists $n{\geq} 0$ such that for all $j{\geq} n$ we have
    that $\sem{r_{j}} = \sem{x {>} C_x}$ and there exists $k > j$ such
    that $F(m_{k}) > 0$;
  \item \label{nz-3} 
    There exists $n{\geq} 0$ such that for all $j{\geq} n$ we have that $\sem{r_{j}}
    = \sem{x <c_x}$ and there exists a $k > j$ such that $F(m_{k}) < 0$; 
  \item \label{nz-4} 
    For all $j{\geq} 0$ there exists $k > j$
    s.t. ${r_{j}} \not = {r_{k}}$; or
  \item \label{nz-5} 
    There exists $n{\geq} 0$ and a thick region $r$ such that
    for all $j{\geq} n$ we have that $r_{j} = r$ and there exists $k > j$ such
    that $F(m_{j}).F(m_{k}) < 0$.
  \end{enumerate}
\end{lemma}
Given an LTL formula $\phi$, and a one variable SHA $\Hh$, we can express in LTL
the conditions  characterizing non-zeno runs of $\Hh$ as given by Lemma~\ref{thm:sha-one-var-nz-main}.  
Let $\phi_C$ be this LTL formula. Model checking of $\phi$ over all non-Zeno
runs then reduces to  standard model checking against formula $\phi \wedge \phi_C$.  
   \qed
 \end{proof}






\section{Weak Singular Hybrid Automata}
\label{sec:reach-sched}
We begin this section by formally introducing constant-rate multi-mode systems
and review the decidability of reachability and schedulability problems for this
class. 
We later present the weak singular hybrid automata model and show the decidability
of various verification problems. 
\subsection{Constant-rate Multi-mode Systems}
\begin{definition}[Constant-Rate Multi-Mode Systems]
We say that a singular hybrid automaton $\Hh = (M, M_0, \Sigma, X, \Delta, I,
F)$ is a \emph{constant-rate multi-mode system} if 
\begin{itemize}
  \item 
    there is a bounded and open polytope  $S$, called the safety set, such that for
    all modes $m \in M$ we have that $I(m) = S$, and
  \item 
    all the modes in  $M$ form a strongly-connected-component, and for every
    mode $m, m' \in M$ if there is a transition $(m, G, a, R, m') \in \Delta$
    then $G=\top$, and $R=\emptyset$. 
  \end{itemize}
\end{definition}
We have slightly modified the definition of CMS from~\cite{ATW12}  to adapt it
to the presentation used in this paper. 
Moreover, we have restricted the safety set to be an open set to avoid problems
in reaching a valuation using infinitely many transitions. 
Notice that there is no structure in a CMS in the sense that all of the modes
can be chosen in arbitrary order as long as the safety set is not violated. 
Alur et al.~\cite{ATW12} showed that due to lack of structure, the
schedulability and the reachability problems for CMS can be reduced to LP
feasibility problem, and hence can be solved in polynomial time. 
\begin{theorem}[Reachability and Schedulability for CMS~\cite{ATW12}]
  \label{thm:mms-reach}
  \label{thm:mms-sched}
  The schedulability and the reachability problems for CMS can be solved in
  polynomial time. 
\end{theorem}
\begin{proof}
  \noindent{\bf Reachability}. 
  Let $\Hh = (M, M_0, X, \Delta, I, F)$ be a CMS with the safety set $S$ and 
  the target  polytope $\Tt$ be given as a system of linear inequalities $AX
  \leq \vb$. 
  Moreover assume that $\nu$ and $\Tt$ are in the safety set $S$.
  Alur et al. showed that the target set $\Tt$ is reachable from 
  $\nu$ iff the following linear program is feasible: 
  \begin{eqnarray}
    \nu + \sum_{m \in M} F(m) \cdot t_m &=& \nu', \nonumber\\
    A \nu' &=& \vb, \text{ and } \label{eqnreachmms}\\
    t_m  &>=&  0, \text{ for all $m \in M$}.\nonumber 
  \end{eqnarray}
  If this linear program is not feasible, then it is immediate that it is not
  possible to reach any valuation in $\Tt$  from $\nu$ using modes in any
  sequence from $\Hh$. 
  On the other hand, if the program is feasible, and as long as both the
  starting valuation and the target set are strictly inside the safety set, a
  satisfying assignment $\seq{t_m}_{m \in M}$ can be used to make progress
  towards $\Tt$ by scaling $t_m$'s appropriately without leaving the safety
  set. 
  Since the feasibility of the linear program can be decided in polynomial time,
  it follows that reachability for the CMS can be decided in polynomial time. 

  \noindent{\bf Schedulability}. 
    Let $\Hh = (M, M_0, X, \Delta, I, F)$ be a CMS with the safety set $S$, and
  initial valuation $\nu$.  
  In this case, Alur et al. showed that there exists a non-Zeno run from
  arbitrary valuation in the safety set  if and only if the following linear
  program is feasible:   
  \begin{eqnarray}
    \sum_{m \in M} F(m) \cdot t_m &=& \vzero,\nonumber \\
    \sum_{m \in M} t_m  &=& 1  \text{ and }   ~\label{eqnschedmms}\\
    t_m  &>=& 0, \text{ for all $m \in M$}.\nonumber
  \end{eqnarray}
  If this linear program is not feasible, then by Farkas's lemma it follows that
  there is a vector $\vv$ such that taking any mode for nonnegative time makes
  some progress in the direction of $\vv$. 
  Hence any non-Zeno run will eventually leave the safety set. 
  On the other hand, if the program is feasible, then a 
  satisfying assignment $\seq{t_m}_{m \in M}$ can be scaled down to stay in a
  ball of arbitrary size around the initial valuation. 
  Hence, if the starting valuation is strictly in the interior of the safety
  set, the feasibility of the linear program~(\ref{eqnschedmms}) imply the
  existence of a non-Zeno run. \qed
\end{proof}

\subsection{Syntax and Semantics}
\emph{Weak singular hybrid automata} (WSHA) can be considered as generalized
constant-rate multi-mode systems with structure, and thus bringing the CMS
closer to singular hybrid automata.  
The restriction on WSHA ensures that the strongly connected components of WSHA
form  CMS, and thus recovering the decidability for the reachability and the
schedulability problem.   
Formally we define WSHA in the following manner. 
\begin{definition}[Weak Singular Hybrid Automata]
  A weak singular hybrid automaton $\Hh = (M, M_0, \Sigma, X, \Delta, I, F)$ is
  a SHA with the restriction 
  that there is a partition on the set of modes $M$ characterized by a function
  $\varrho: M \rightarrow \Nat$ assigning \emph{ranks} to the modes such that
  \begin{itemize}
  \item  
    for every transition $(m, G, a, R, m') \in \Delta$ we have that    
    $\varrho(m) \leq \varrho(m')$, and 
  \item  
    for every rank $i$ the set of modes $M_i = \set{m \::\: \varrho(m) = i}$ is
    such that 
    \begin{itemize}
    \item[--] 
      there is a bounded and open polytope $S_i$, called the safety set of $M_i$,
      such that  for all modes $m \in M_i$ we have that $I(m) = S_i$; and
    \item[--]
      all the modes in $M_i$ form a strongly-connected-component, and for every
      mode $m, m' \in M_i$ if there is a transition $(m, G, a, R, m') \in \Delta$
      then $G=\top$, and $R=\emptyset$. 
    \end{itemize}
  \end{itemize}
\end{definition}
Observe that every CMS is a weak singular hybrid automaton (WSHA), and every
strongly connected component of a WSHA is a CMS. 
Also notice that for every (finite or infinite) run 
$r = \seq{(m_0, \nu_0), (t_1, a_1), (m_1, \nu_1), \ldots}$
of a WSHA we have that $\varrho(m_i) \leq \varrho(m_j)$ for every $i \leq j$.
We define the type $\Gamma(r)$ of a finite run 
$r = \seq{(m_0, \nu_0), (t_1, a_1), (m_1, \nu_1), \ldots, (m_k, \nu_k)}$
as a finite sequence of ranks (natural numbers) and actions  
$\seq{n_0, b_1,  n_1,\ldots, b_p, n_p}$ defined inductively in the following manner: 
 \begin{eqnarray*}
   \Gamma(r) = 
 \begin{cases}
   \seq{\varrho(m_0)} & \text{ if $r = \seq{(m_0, \nu_0)}$}\\
   \Gamma(r') \oplus (a, \varrho(m)) & \text{ if $r = r'::\seq{(t, a), (m, \nu)}$},
 \end{cases}
 \end{eqnarray*}  
where $::$ is the cons operator that appends two sequences, while 
for a sequence $\sigma = \seq{n_0, b_1, n_1, \ldots, b_p, n_p}$, $a \in \Sigma$, and $n \in
\Nat$ we define $\sigma \oplus (a, n)$ to be equal to $\sigma$ if $n_p = n$ and
$\seq{n_0, b_1, n_1, \ldots, n_p, a, n}$ otherwise.
Intuitively, the type of a finite run gives the (non-duplicate) sequence of
ranks of modes and actions appearing in the run, where action is stored only
when a transition to a mode of higher rank happens. 
We need to remember only these actions since transitions that stay in the modes
of same rank do not reset the variables. 
It is an easy observation that, since there are only finitely many ranks for a
given WSHA, we have that for every infinite run 
$r = \seq{(m_0, \nu_0), (t_1, a_1), (m_1, \nu_1), \ldots}$ there
exists an index $i$ such that  for all $j \geq i$ we have that $\varrho(m_i) =
\varrho(m_j)$. 
With this intuition we define the type of an infinite run $r$ as the type of the
finite prefix of $r$ till index $i$.
We write $\Gamma_\Hh$ for the set of run types of a WSHA $\Hh$.
\begin{theorem}
  \label{thm:wsha-reach-np-hard}
  The reachability and the schedulability problems for weak singular hybrid
  automata is NP-complete.
\end{theorem}
\begin{proof}{(Sketch)}
  To show NP-membership we show that to decide the reachability problem, it is
  sufficient to guess a finite run type, and check whether there is a run with
  that type that reaches the target polytope. 
  Since the size of every run type is polynomial in the size of the WSHA, and
  there are only exponentially many run-types, if for a run we can check whether
  there exists a run of this type reaching target polytope is polynomial time,
  the NP-membership claim follows. 
  Given a run type $\sigma = \seq{n_0, b_1, n_1, \ldots, b_p, n_p}$ an initial
  valuation $\nu_0$ and a bounded and convex target polytope $\Tt$ given as $A X
  \leq \vb$, there exists a run with type $\sigma$ that reaches a valuation in
  $\Tt$ if and only if the following linear program is feasible: 
  for every $ 0 \leq i \leq p$ and  $m \in M_{n_i}$ there are 
  $\nu_{n_i}, \nu_{n_i}' \in \Real^{|X|}$ and  $t_i^m \in \Rplus$ such that:
  \begin{eqnarray}
    \nu_0 &=& \nu_{n_0}, \nu_{n_p}'  \in \Tt \nonumber\\
    \nu_{n_i}, \nu_{n_i}'  & \in & S_{M_{n_i}} \text{ for all $0 \leq i \leq p$} 
    \nonumber \\
    \nu_{n_i}   & \in & G(b_i) \text{ for all $0 < i \leq p$}
    \nonumber \\
    \nu_{n_{i+1}}(j) &=& 0 \text{ for all $x_j \in R(b_{i+1})$ and  $0 < i \leq p$}
    \nonumber \\
    \nu_{n_{i+1}}(j) &=& \nu_{n_i}'(j)  \text{ for all $x_j \not \in R(b_{i+1})$ and  $0 < i \leq p$}
    \nonumber \\
    \nu_{n_i}'  & = & \nu_{n_i} + \sum_{m \in M_{n_i}} F(m) \cdot t_i^m \text{
      for all $0 \leq i \leq p$}\nonumber\\
    t_i^m  &\geq&  0 \text{ for all $0 \leq i \leq p$ and  $m \in M_{n_i}$} \nonumber
  \end{eqnarray}
  These constraints check whether it is possible to reach some valuation in the
  target polytope while satisfying the guard and constraints of the WSHA, while
  exploiting the fact that modes of same rank can be applied an arbitrary number
  of time in an arbitrary order.
  The proof for this claim is similar to the proof for the CMS, and hence
  omitted. 

To show NP-hardness we reduce the \emph{subset-sum problem} to solving the
reachability problem in a WSHA.  
Formally, given $A$, a non-empty set of $n$  integers and another integer $k$,
the \emph{subset-sum problem} is to determine if there is a non-empty subset $T
\subseteq A$ that sums to $k$. 
Given the set $A$ and the integer $k$, we construct a WSHA $\Hh$ with $n+3$
variables $x_0,x_1, \dots, x_{n+2}$, $2n+1$ modes $m_0, \dots, m_{2n}$ and
$2n$ transitions, such that starting from a given valuation, a particular 
target polytope $\Tt$ is reachable in
the WSHA iff there is a non-empty subset $T \subseteq A$ that sums up to $k$.
Intuitively, the variable $x_0$ ensures that the variables
$x_1, x_2\ldots x_n$ are initialized with values $a_1, a_2, \ldots a_n$ (the elements
of $A$). The variable $x_{n+1}$ sums up the values of the elements in the chosen
subset $T$, and can be later compared with $k$.
The variable $x_{n+2}$ ensures that the set $T$ is non-empty (specifically when 
$k = 0$).
The rates in the modes $m_i (0\leq i\leq 2n)$ are given as follows ($\vr_i$ represents
$(F(m_i))$:
\begin{itemize}
\item[$\bullet$] $\vr_0(x_0) =1, \vr_0(x_{n+1}) = \vr_0(x_{n+2}) = 0$, and
$\vr_0(x_i) = a_i$, where $1\leq i\leq n$
\item[$\bullet$] $\vr_{2j-1}(x_j) = -a_j$, $\vr_{2j-1}(x_{n+1}) = a_j$,
$\vr_{2j-1}(x_{n+2}) = 1$, and $\vr_{2j-1}(x_0)$ = $\vr_{2j-1}(x_i) = 0$,
where $1\leq i\neq j\leq n$
\item[$\bullet$] $\vr_{2j}(x_j) = -a_j$, and $\vr_{2j}(x_0) = \vr_{2j}(x_i) =$
$\vr_{2j}(x_{n+1}) = \vr_{2j}(x_{n+2}) = 0$,
 where $1\leq i\neq j\leq n$
\end{itemize}
The transitions are as follows:
(i) There are edges from $m_0$ to $m_1$ and to $m_2$, and (ii)
 There are edges from $m_{2j-1}$ and from $m_2j$ to $m_{2j+1}$ and to $m_{2j+2}$, $1\leq j\leq n-1$. 
We  claim that the polytope $\Tt$, given by the set of points $\px$
such that $\px(0) = 1$, $\px(i) = 0$, $1\leq i\leq n$, $\px(n+1) = k$
and $\px(n+2) \in [1,n]$, is reachable from the point $\vzero$
iff there is a non-empty subset $T$ of $A$ that sums up to $k$.
Figure \ref{nphard-main} gives an illustration of the WSHA construction for a
set $\{1,2,-3\}$.  
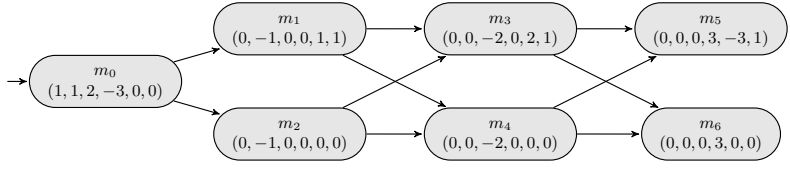
\begin{figure}[t]
\begin{center}
  \scalebox{0.7}{
\begin{tikzpicture}[->,>=stealth',shorten >=1pt,auto,node distance=1.8cm,
  semithick]
  \tikzstyle{every state}=[fill=black!10!white,minimum size=3em,rounded rectangle]
  \node[initial,state,initial where = left, initial text={}] at (0.5, 0) (A0) {$\begin{array}{c}m_0 \\ (1,1,2,-3,0,0) \end{array}$} ;
  \node[state] at (4, 1) (A1) {$\begin{array}{c}m_1 \\ (0,-1,0,0,1,1) \end{array}$} ;
  \node[state] at (4, -1) (A) {$\begin{array}{c} m_2 \\ (0,-1,0,0,0,0) \end{array}$} ;
  \node[state] at (8, 1) (B1) {$\begin{array}{c}m_3 \\ (0,0,-2,0,2,1) \end{array}$} ;
  \node[state] at (8 , -1) (B) {$\begin{array}{c} m_4 \\ (0,0,-2,0,0,0) \end{array}$} ;
  \node[state] at (12, 1) (B2) {$\begin{array}{c}m_5 \\ (0,0,0,3,-3,
      1) \end{array}$} ;
  \node[state] at (12, -1) (C) {$\begin{array}{c} m_6 \\ (0,0,0,3,0,0) \end{array}$} ;
  
      \path (A0) edge node [above]{} 
                    (A);
                    \path (A0) edge node [above]{} 
                    (A1);
      \path (A1) edge node [above]{} 
                    (B);
                    \path (A1) edge node [above]{} 
                    (B1);
     \path (A) edge node [above]{} 
                    (B);
                    \path (A) edge node [above]{} 
                    (B1);
     \path (B1) edge node [above]{} 
                    (B2);
                    \path (B1) edge node [above]{} 
                    (C);
\path (B) edge node [above]{} 
                    (B2);
       \path (B) edge node [above]{} 
                    (C);
 \end{tikzpicture}
 }
\caption{Constructed WSHA for a set $\{1,2,-3\}$}
\vspace{-2em}
\label{nphard-main}
\end{center}
\end{figure}
The detailed proof can be found in the appendix.
Due to the lack of space, the proof for the schedulability is moved to
  the appendix.
\qed 
\end{proof}

\begin{corollary}
  \label{thm:wsha-ltl-mc}
  The LTL model-checking problem for WSHA is PSPACE-complete.
\end{corollary}

We also observe that CTL model checking for weak singular hybrid automata is
already hard for \PSPACE{} by using a reduction from subset sum games~\cite{FJ13}.
\begin{theorem}
   \label{thm:ctl-weak-sha}
   CTL model checking of weak SHAs with two clock variables is \PSPACE{}-hard. 
 \end{theorem}
\begin{proof}[Sketch]
  We give a polynomial reduction from subset-sum games. 
  A subsetsum game is played between an existential player and a universal
  player. 
  The game is specified by a pair $(\psi, T)$ where $T \in \N$ and $\psi$ is a list:
  $$\forall\{A_1,B_1\} \exists\{E_1,F_1\} \dots \forall \{A_n,B_n\}\exists \{E_n,F_n\}$$
  where $A_i, B_i, E_i, F_i$ are all natural numbers. 
  The game is played in rounds. 
  In the first round,  the universal player chooses an element from
  $\{A_1,B_1\}$, and the existential player responds by choosing a number from
  $\{E_1,F_1\}$. 
  In the next round, the universal player chooses an element from $\{A_2,B_2\}$,
  and the existential player responds  by choosing a number from $\{E_2,F_2\}$. 
  This pattern repeats for $n$ rounds, and two players this construct a sequence
  of number, and the existential player wins iff the sum of those numbers equals $T$. 
 
  For each set $\{A_i,B_i\}$, we construct a widget $W_{\forall_i}$ shown in the
  left side of Figure~ \ref{exists}.
  Similarly, for each set $\{E_i,F_i\}$, we construct a widget $W_{\exists_i}$
  shown in the right side of Figure~\ref{exists}.
  \begin{figure}[b]
    \begin{center}
      \begin{tikzpicture}[->,>=stealth',shorten >=1pt,auto,node distance=1.8cm,
        semithick]
        \tikzstyle{every state}=[fill=black!30!white,minimum size=3em,rounded rectangle]
        \node[state] (A1) {$Q_i$} ;
        \node[state] at (4,0) (C) {$P_i$} ;
        \path (A1) edge [bend left=30] node [above]{$x=A_i?$} 
        node [below]{$x:=0$}(C);
        \path (A1) edge [bend right=30]node  [above]{$x=B_i?$} 
        node [below]{$x:=0$}(C);
      \end{tikzpicture}
      \hfill
      \begin{tikzpicture}[->,>=stealth',shorten >=1pt,auto,node distance=1.8cm,
        semithick]
        \tikzstyle{every state}=[fill=black!30!white,minimum size=3em,rounded rectangle]
        \node[state] (A1) {$P_i$} ;
        \node[state] at (4,0) (C) {$Q_{i+1}$} ;
        \path (A1) edge [bend left=30] node [above]{$x=E_i?$} 
        node [below]{$x:=0$}(C);
        \path (A1) edge [bend right=30]node  [above]{$x=F_i?$} 
        node [below]{$x:=0$}(C);
      \end{tikzpicture}
      \caption{Widgets $W_{\forall_i}$ (left) and $W_{\exists_i} (right)$}
      \label{exists}
    \end{center}
  \end{figure}
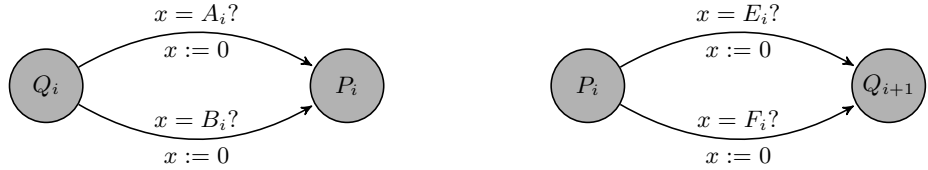
  
  The WSHA ${\cal A}$ constructed has 2 clocks $x,y$ and 
  is obtained by connecting  the last node of widget $W_{\forall_i}$ with 
  the starting node of widget $W_{\exists_i}$ for $1 \leq i \leq n$, and 
  by connecting the last node of $W_{\exists_i}$ with the initial node of
  $W_{\forall_{i+1}}$, $1 \leq i \leq n-1$. 
  The last node of $W_{\exists_n}$ (labeled $P_n$) is connected to a node
  $\textsc{End}$ with a guard $y=T$. 
  The clock $y$ is never reset in any of the widgets and accumulates the
  computed sum.  
  Notice that the resulting timed automaton is a WSHA since it is acyclic. 
  Let $i_0$ be the initial mode of this WSHA. 
  The unique initial mode of the WSHA is labeled with $Q_1$. 
  It is easy to see that the above WSHA can be constructed 
  in polynomial time; moreover, the constructed WSHA has 
  $3n+2$ modes, $4n+1$ edges and 2 clocks. 
  We now give a CTL formula $\varphi$ of size $\mathcal{O}(n)$ given by 
  $[Q_1 \wedge \A\bigcirc(P_1 \wedge \E\bigcirc (Q_2 \wedge \dots \A\bigcirc(P_n
  \wedge \E\bigcirc \textsc{End})))]$   
  It can be seen that the existential player wins the subsetsum game iff 
  ${\cal A},i_0 \models \varphi$.
\qed
\end{proof}

 We conjecture that the problem can be solved in \PSPACE{}, however decidability
 of the problem is currently open.

\section{Undecidable variants of WSHA}
\label{sec:undec}
In this section, we present two variants  of WSHA and show that both lead to
undecidability of the reachability problem. 

\begin{theorem}
  \label{thm:sha-undec-three-var}
  The reachability problem is undecidable for three variable WSHAs with discrete
  updates. \todo{What is meant by "Discrete updates"?}
 \end{theorem}
\begin{proof}[Sketch]
\todo{the proof of Th.7 (even in appendix) does not justify that the constructed
automaton is a WSHA}
\todo{how "WSHAs with discrete
updates" are really a subclass of standard singular HA}
  We show a reduction from the halting problem for two-counter Minsky machines
  ${\cal M}$. 
  The variables $x,y,z$ of the SHA  have global invariants $0 \leq x \leq 1$, $0
  \leq y,z\leq 5$ respectively. 
  The counters $c_1, c_2$ of two-counter machine are encoded in variables $y$ and
  $z$ as 
  $  y=5-\frac{1}{2^{c_1}} \text{ and  } z=5-\frac{1}{2^{c_2}}$.
  To begin, we have $c_1=c_2=0$, hence $y=z=4$, and $x{=}0$.
  The rates of $x,y,z$ are indicated by 3-tuples inside the modes of the SHA. 
  The discrete updates on $x,y,z$ are indicated on the transitions.
  We construct widgets for each of the increment/decrement and zero check instructions. 
  Each widget begins with $x{=}0,y=5-\frac{1}{2^{c_1}}$ and
  $z=5-\frac{1}{2^{c_2}}$, where $c_1,c_2$ are the counter values. 
  \begin{itemize}
  \item {\bf (Increment and Decrement Instructions).}
    Let us first consider increment instruction 
    $    l : c_1 := c_1 + 1 \text{ goto } l'$.
    The Figure~\ref{inc-c1-main} depicts the increment widget. 
    This widget starts with a mode labeled $l$, and  ends in a mode labeled $l'$. 
    This widget can be modified to simulate the instructions increment counter $c_2$,
    decrement counters $c_1$ and decrement counter $c_2$,respectively,  by changing  the cost rate of $z$ at
    $B$ to -3, the cost rate of $y$ at $B$ to -12, and the cost rate of $z$ at $B$
    to $-12$, respectively.  
    
  \item {\bf (Zero Check Instruction).}
    We next consider the zero check instruction:
    $l: \text{ if } c_1=0 \text{ goto } l' \text{ else  goto } l''$.
    The widget of Figure~\ref{z-c1-main} depicts the zero check widget. 
    The zero check widget  starts in a mode $l$ and reaches either the mode $l'$ or mode $l''$.
    
\end{itemize}
 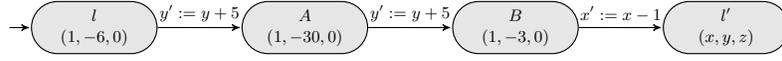
\begin{figure}[t]
  \begin{center}
  \scalebox{0.7}{
  \begin{tikzpicture}[->,>=stealth',shorten >=1pt,auto,node distance=1.8cm,
    semithick]
    \tikzstyle{every state}=[fill=black!10!white,minimum width=8em,rounded rectangle]
    \node[initial,state, initial text={}] at (-4,0) (A1) {$\begin{array}{c}l \\ (1,-6,0)\end{array}$} ;
    \node[state] at (0,0) (A) {$\begin{array}{c}A \\ (1,-30,0)\end{array}$} ;
    \node[state] at (4,0) (B) {$\begin{array}{c}B \\ (1,-3,0)\end{array}$} ;
    \node[state] at (8,0) (C) {$\begin{array}{c}l' \\ (x,y,z)\end{array}$} ;
    \path (A1) edge node [above]{$y' := y + 5$}(A);
    \path (A) edge node [above]{$y' := y + 5$}(B);
    \path (B) edge node [above]{$x' := x - 1$}(C); 
  \end{tikzpicture}
  }
  \caption{Increment $c_1$ widget}
\vspace{-2em}
   \label{inc-c1-main}
  \end{center}
\end{figure}

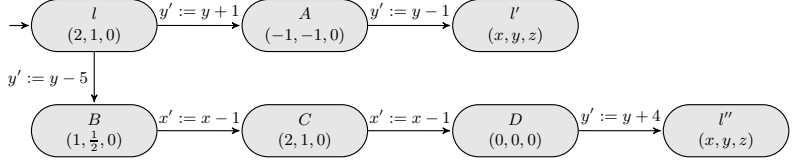
\begin{figure}[h]
  \begin{center}
  \scalebox{0.7}{
  \begin{tikzpicture}[->,>=stealth',shorten >=1pt,auto,node distance=1.8cm,
    semithick]
    \tikzstyle{every state}=[fill=black!10!white,minimum width=8em,rounded rectangle]
    
    \node[initial,state, initial text={}] at (0,0) (A1) {$\begin{array}{c}l \\ (2,1,0)\end{array}$} ;
    \node[state] at (4,0) (A) {$\begin{array}{c}A \\ (-1,-1,0)\end{array}$} ;
    \node[state] at (8,0) (A2) {$\begin{array}{c}l' \\ (x,y,z)\end{array}$} ;

    \node[state] at (0,-2) (E) {$\begin{array}{c}B \\(1,\frac{1}{2},0) \end{array}$} ;
    \node[state] at (4,-2) (B) {$\begin{array}{c}C \\(2,1,0) \end{array}$} ;
    \node[state] at (8,-2) (B2) {$\begin{array}{c}D \\ (0,0,0)\end{array}$} ;
    \node[state] at (12,-2) (C2) {$\begin{array}{c}l'' \\ (x,y,z)\end{array}$} ;
    
    \path (A1) edge node [above]{$y' := y + 1$}(A);
    \path (A) edge node [above]{$y' := y - 1$}(A2);
    \path (B) edge node [above]{$x' := x - 1$}(B2);
    \path (B2) edge node [above]{$y' := y + 4$}(C2);
    \path (A1) edge node [left]{$y' := y - 5$}(E);
    \path (E) edge node [above]{$x' := x - 1$}(B);
  \end{tikzpicture}
}
\caption{Zero Check widget}
\vspace{-2em}
\label{z-c1-main}
\end{center}
\end{figure}
It is straightforward to see that the modules for increment, decrement, and zero
check simulate the two counter machine. There is a mode HALT corresponding to the HALT  instruction. 
The halting problem for two counter machines is thus  reduced to the
reachability of the mode HALT. 
\qed
\end{proof} 

The next variant that we consider is weak singular hybrid automata
extended with discrete updates on the variables even inside the strongly
connected components.\todo{Discrete updates were there in the last theorem}
We prove this result by showing that even CMS with one unrestricted
clock variable lead to undecidability.  
 \begin{theorem}
   \label{thm:sha-undec-three-var-one-clock}
 The reachability problem is undecidable for CMS with three variables and one
unrestricted clock. 
 \end{theorem}
 \begin{proof}[Sketch]
  We simulate a two counter machine using a CMS with 3 variables and one clock. 
 The three variables $x_1,x_2,y$ have global invariants $0 \leq x_1, x_2 \leq 5$ and $0 \leq y \leq 1$ 
 respectively. The clock variable is $x$. 
 The counters $c_1,c_2$ are encoded as $x_1=5-\frac{1}{2^{c_1}}$, $x_2=5-\frac{1}{2^{c_2}}$. 
 At the beginning of each widget, we have $x_1=5-\frac{1}{2^{c_1}}$, $x_2=5-\frac{1}{2^{c_2}}$ and $y=1$, where 
 $c_1, c_2$ are the current counter values.
 The gadget simulating an increment instruction 
 is shown below. 
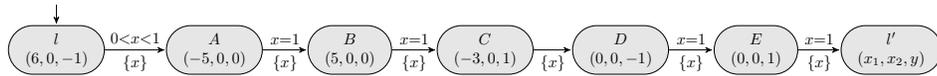
\begin{center}
 \scalebox{0.6}{
   \begin{tikzpicture}[->,>=stealth',shorten >=1pt,auto,node distance=1.8cm, semithick]
     \tikzstyle{every state}=[fill=black!10!white,minimum size=0em,rounded rectangle]
     \node[initial, initial where=above,state, initial text={}] at (-10.5,0) (A1) {$\begin{array}{c}l\\(6,0,-1)\end{array}$} ;
     \node[state] at (-7,0) (A) {$\begin{array}{c}A \\ (-5,0,0)\end{array}$} ;
     \node[state] at (-4,0) (B) {$\begin{array}{c}B \\ (5,0,0)\end{array}$} ;
   \node[state] at (-1,0) (C) {$\begin{array}{c}C \\ (-3,0,1)\end{array}$} ;
   \node[state] at (2,0) (D) {$\begin{array}{c}D \\ (0,0,-1)\end{array}$} ;
    \node[state] at (5,0) (E) {$\begin{array}{c}E \\ (0,0,1)\end{array}$} ;
      \node[state] at (8,0) (F) {$\begin{array}{c}l' \\ (x_1,x_2,y)\end{array}$} ;
    \path (A1) edge node [above]{$0{<}x{<}1$}
		    node [below]{$\set{x}$}(A);
    \path (A) edge node [above]{$x{=}1$}
                    node [below]{$\set{x}$}(B);
\path (B) edge node [above]{$x{=}1$}
                    node [below]{$\set{x}$}(C);
\path (C) edge  node [below]{$\set{x}$}(D);
\path (D) edge node [above]{$x{=}1$}
                    node [below]{$\set{x}$}(E);
                    
   \path (E) edge node [above]{$x{=}1$}
                    node [below]{$\set{x}$}(F);
                 
  \end{tikzpicture}
  }
\end{center}
 
The decrement gadget is similar to the increment gadget.   
The gadget for a zero check 
is given 
in the figure below.
Observe that starting at mode $l$ with $y=1$ and $x_1=5-\frac{1}{2^{c_1}}$,    
the gadget in this gadget ensures that we reach $l'$ iff $c_1=0$,
and otherwise reaches $l''$. 
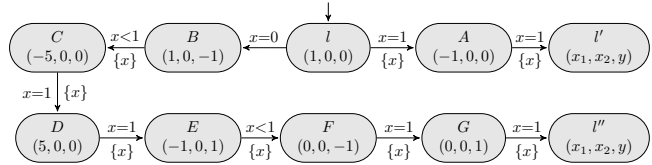
\begin{center}
 \scalebox{0.6}{
\begin{tikzpicture}[->,>=stealth',shorten >=1pt,auto,node distance=1.8cm,
  semithick]
\tikzstyle{every state}=[fill=black!10!white,minimum size=3em,rounded rectangle]
 \node[initial,state,initial where = above, initial text={}] at (0,0) (A1) {$\begin{array}{c}l \\ (1,0,0)\end{array}$} ;
  \node[state] at (3,0) (A) {$\begin{array}{c}A \\ (-1,0,0)\end{array}$} ;
   \node[state] at (6,0) (A2) {$\begin{array}{c}l' \\ (x_1,x_2,y)\end{array}$} ;
  \node[state] at (-3,0) (E) {$\begin{array}{c}B \\(1,0,-1) \end{array}$} ;
    \node[state] at (-6,0) (B) {$\begin{array}{c}C \\(-5,0,0) \end{array}$} ;
   \node[state] at (-6,-2) (B2) {$\begin{array}{c}D \\ (5,0,0)\end{array}$} ;
   \node[state] at (-3,-2) (C2) {$\begin{array}{c}E \\ (-1,0,1)\end{array}$} ;
    \node[state] at (0,-2) (F) {$\begin{array}{c}F \\ (0,0,-1)\end{array}$} ;
     \node[state] at (3,-2) (G) {$\begin{array}{c}G \\ (0,0,1)\end{array}$} ;
     \node[state] at (6,-2) (H) {$\begin{array}{c}l''\\ (x_1,x_2,y)\end{array}$} ;
    \path (A1) edge node [above]{$x{=}1$}
                         node [below]{$\set{x}$}(A);
\path (A) edge node [above]{$x{=}1$}
                         node [below]{$\set{x}$}(A2);
   \path (A1) edge node [above]{$x{=}0$}(E);
\path (E) edge node [above]{$x{<}1$}
                         node [below]{$\set{x}$}(B);
\path (B) edge node [left]{$x{=}1$}
                         node [right]{$\set{x}$}(B2);
\path (B2) edge node [above]{$x{=}1$}
                         node [below]{$\set{x}$}(C2);
\path (C2) edge node [above]{$x{<}1$}
                         node [below]{$\set{x}$}(F);
\path (F) edge node [above]{$x{=}1$}
                         node [below]{$\set{x}$}(G);
\path (G) edge node [above]{$x{=}1$}
                         node [below]{$\set{x}$}(H);
\end{tikzpicture}
}\end{center}
A complete proof can be found in the Appendix. 
\qed\end{proof}

\section{Conclusion}
\label{sec:concl}
We introduced weak singular hybrid automata and showed that verification
problems like reachability and schedulability are $\NPC{}$, while LTL
property checking is $\PSPACEC{}$.     
Extending the model with either unrestricted variable updates or with a single
unrestricted clock variable render the reachablity problem undecidable.  
We showed PSPACE-hardness of the CTL model checking problem, but the exact
complexity of the problem remains open.

\bibliographystyle{plain}
\bibliography{papers}

\newpage
\appendix
\section{Temporal Logic Model Checking} 
Model-checking---pioneered by Clarke, Sifakis and Emerson~\cite{CES86}---is 
widely used automated verification framework that, given a formal description of
a system and a property, systematically checks whether this property holds for a
given state of the system model.  

The linear temporal logic, LTL, and the computational tree logic (CTL) provide
formal languages to specify more involved nesting of such properties with ease.  

The first step in introducing logic to specify the properties of an SHA is to
specify properties of interest as propositions. 
A more general way to assign propositions with an SHA is to introduce
propositions on the valuations of the variables. 
However, in this paper, we consider the propositions given on modes of the SHA.   
 
A Kripke singular hybrid automaton (KSHA) is a tuple $(\Hh, \Pp, L)$ where
$\Hh$ is an SHA,  $\Pp$ is a finite set of atomic propositions, and $L: M \to
2^\Pp$ is a labeling function that labels the modes of the SHA with a subset of
atomic propositions $\Pp$. 
Given a KSHA $(\Hh, \Pp, L)$ and an infinite run 
$r=\seq{(m_0, \nu_0), (t_1, a_1), (m_1, \nu_1), \ldots}$
of $\Hh$, we define a trace corresponding to $r$, denoted as 
$\Trace(r)$, as the sequence 
\[
\seq{L(m_0), L(m_1), L(m_2),  \dots L(m_n), \ldots}.
\]
Let $\Trace(\Hh, \Pp, L)$ be the set of traces of the KSHA $\Hh$.
For a trace $\sigma = \seq{P_0, P_1, \ldots, P_n, \ldots}  \in \Trace(\Hh, P,
L)$ we write $\sigma[i]=\seq{P_i, P_{i+1}, \ldots}$ for the suffix of the trace
starting at the index $i \geq 0$. 

We now define the syntax and semantics of LTL and CTL. 

\subsection{Linear Temporal Logic}
\begin{definition}[Linear Temporal Logic (Syntax)]
The set of valid LTL formulas over a set $\Pp$ of atomic propositions can be
  inductively defined as the following:
\begin{itemize}
\item $\top$ and $\bot$ are valid LTL formulas;
\item if $p \in \Pp$ then $p$ is a valid LTL formula;
\item if $\phi$ and $\psi$ are valid LTL formulas then so are $\neg \phi$, $\phi
  \wedge \psi$ and $\phi \lor \psi$; 
\item if $\phi$ and $\psi$ are valid LTL formulas then so are $\bigcirc\phi$,
  $\Diamond \phi$, $\square \phi$, and $\phi \until \psi$.
\end{itemize}
\end{definition}

We often use $\phi \Rightarrow \psi$ as a shorthand for $\neg \phi \lor \psi$.
Before we define the semantics of LTL formula formally, let us give an informal
description of the temporal operators $\bigcirc$, $\Diamond$, $\square$, and
$\until$. 
LTL formulas are interpreted over traces of KSHA. 
The formula $\bigcirc \phi$, read as next $\phi$, holds for a trace
$\sigma=\seq{P_0, P_1, P_2, \ldots}$ if $\psi$ holds for the trace $\sigma[1]$. 
The formula $\Diamond \phi$, read as eventually $\phi$, holds for a trace
$\sigma=\seq{P_0, P_1, P_2, \ldots}$ if there exists $i \geq 0$ such that the
formula $\psi$ holds for the trace $\sigma[i]$. 
The formula $\square \phi$, read as globally or always $\phi$, holds for a trace
$\sigma=\seq{P_0, P_1, P_2, \ldots}$ if for all $i \geq 0$ the formula $\psi$
holds for traces $\sigma[i]$. 
Finally, the formula $\phi \until \psi$, read as $\phi$ until $\psi$, holds for
a trace $\sigma=\seq{P_0, P_1, P_2, \ldots}$ if there is an index $i$ such that
$\psi$ holds for the trace $\sigma[i]$, and for every index $j$ before $i$ the
formula $\phi$ holds for the trace $\sigma[j]$, i.e the formula $\phi$ holds until
formula $\psi$ holds.
\begin{definition}[LTL Semantics]
 For a trace \\ $\sigma$=$\seq{P_0, P_1, P_2, \ldots}$ of a KSHA we write $\sigma
  \models \phi$ to say that the trace 
  $\sigma$ satisfies the formula $\phi$. 
  The satisfaction of LTL formulas is defined as follows:
  \begin{itemize}
  \item 
    $\sigma \models \top$ and $\sigma \not \models \bot$;
  \item 
    $\sigma \models p$ if $p \in P_0$;
  \item
    $\sigma \models \neg \phi$ if $\sigma \not \models \phi$;
  \item
    $\sigma \models \phi \wedge \psi$ if $\sigma \models \phi$ and $\sigma
    \models \psi$;
  \item
    $\sigma \models \phi \lor \psi$ if $\sigma \models \phi$ or $\sigma
    \models \psi$;    
  \item
    $\sigma \models \bigcirc\phi$ if $\sigma[1] \models \phi$;
  \item
    $\sigma \models \Diamond \phi$ if there exists $i \geq 0$ such that $\sigma[i]
    \models \phi$; 
  \item
    $\sigma \models \square \phi$ if for all $i \geq 0$ we have that $\sigma[i]
    \models \phi$; and
  \item
    $\sigma \models \phi \until \psi$ if there exists $i \geq 0$ such that
    $\sigma[i] \models \psi$, and for all $0 \leq j < i$ 
    $\sigma[j]\models \phi$. 
  \end{itemize}
\end{definition}

We consider only those runs of the KSHA that are non-Zeno, for satisfaction/refutation 
  of LTL formule.   For a Kripke singular hybrid automaton  $(\Hh, P, L)$, and an LTL formula $\phi$ we
  say that KSHA $(\Hh, P, L)$ satisfies LTL formula $\phi$, and we write $(\Hh,
  P, L) \models \phi$,  if for all traces of non-Zeno run  $\sigma \in \Trace(\Hh, P, L)$
  we have that $\sigma \models \phi$.






\subsection{Computational Tree Logic}
In this section we present CTL logic to reason with Kripke singular hybrid
automata. 
CTL formulas allows existential and universal quantification over properties of
trace originating from modes. 
Formally, the set of CTL formulas are defined in the following manner. 
\begin{definition}[Computational Tree Logic (Syntax)]
  The set of valid CTL formulas over a set $\Pp$ of atomic propositions can be
  inductively defined as the following:
  \begin{itemize}
  \item 
    $\top$, $\bot$ and $p \in \Pp$  are valid CTL mode formulae;
  \item 
    if $\phi$ and $\psi$ are valid CTL mode formulae then so are $\neg \phi$, $\phi
    \wedge \psi$, $\E \phi$  and $\A \phi$;
  \item 
    if $\phi$ and $\psi$ are valid mode formulae, then  $\bigcirc \phi$ and
    $\phi \until \psi$ are valid trace formulae.  
\end{itemize}
\end{definition}

\begin{definition}[CTL Semantics]
  CTL formulae are evaluated over the modes and traces of a Kripke Hybrid structure.
  Given a mode $m \in M$, 
  \begin{itemize}
  \item $m \models \top$ and $m\not \models \bot$;
  \item $m \models p$ if $p \in L(m)$;
  \item
    $m \models \neg \phi$ if $m \not \models \phi$;
  \item
    $m \models \phi \wedge \psi$ if $m \models \phi$ and $m
    \models \psi$;
  \item
    $m \models \E \phi$ if there exists a trace $\sigma$ starting at $m$ such
    that $\sigma \models \phi$ 
  \item
    $m \models \A \phi$ if all traces $\sigma$ starting at $m$ are such that
    $\sigma \models \phi$ 
  \item 
    $\sigma \models \bigcirc \phi$ iff $\sigma[1] \models \phi$
  \item 
    $\sigma \models \phi \until \psi$ if there exists $i \geq 0$ such that
    $\sigma[i] \models \psi$, and for all $\sigma[j]\models \phi$ $0 \leq j < i$. 
  \end{itemize}
  
\end{definition}

We consider only traces of non-Zeno runs of the  KSHA for
  satisfaction/refutation of CTL formule.  
  Given a CTL formula $\phi$ and a KSHA $(\Hh, \Pp, L)$, we say that 
  $(\Hh, \Pp, L) \models \phi$ if and only if all initial modes $m_0$ of ${\mathcal A}$
  satisfy $\phi$.
%



\section{Two-Counter Machines}

A two-counter machine (Minsky machine) $\Aa$ is a tuple $(L, C)$ where:
${L = \set{l_0, l_1, \ldots, l_n}}$ is the set of
instructions. There is a distinguished terminal instruction  $l_n$ called
HALT. 
${C = \set{c_1, c_2}}$ is the set of two \emph{counters};
the instructions $L$ are one of the following types:
\begin{enumerate}
\item (increment $c$) $l_i : c := c+1$;  goto  $l_k$,
\item (decrement $c$) $l_i : c := c-1$;  goto  $l_k$,
\item (zero-check $c$) $l_i$ : if $(c >0)$ then goto $l_k$
  else goto $l_m$,
\item (Halt) $l_n:$ HALT.
\end{enumerate}
where $c \in C$, $l_i, l_k, l_m \in L$.

A configuration of a two-counter machine is a tuple $(l, c, d)$ where
$l \in L$ is an instruction, and $c, d$ are natural numbers that specify the value
of counters $c_1$ and $c_2$, respectively.
The initial configuration is $(l_0, 0, 0)$.
A run of a two-counter machine is a (finite or infinite) sequence of
configurations $\seq{k_0, k_1, \ldots}$ where $k_0$ is the initial
configuration, and the relation between subsequent configurations is
governed by transitions between respective instructions.
The run is a finite sequence if and only if the last configuration is
the terminal instruction $l_n$.
Note that a two-counter  machine has exactly one run starting from the initial
configuration. 
The \emph{halting problem} for a two-counter machine asks whether 
its unique run ends at the terminal instruction $l_n$.
It is well known~(\cite{Min67}) that the halting problem for
two-counter machines is undecidable.

\section{Omitted Proofs}
\subsection{\underline{Proofs from Section~\ref{sec:definitions}}}
\subsubsection{Proof of Theorem~\ref{thm:ctl-mc-undec}}
  Given a two counter machine ${\cal M}$, we construct a two variable SHA ${\cal T}$ and a CTL formula 
 $\varphi$ such that ${\cal T} \models \varphi$ iff ${\cal M}$ halts. The counter values are stored 
 in a variable $x$ in the form $2-\frac{1}{2^{c_1}3^{c_2}}$, while $y$ is an extra variable used for manipulations.  
The rate of $x$ is given inside the modes of the SHA, and the rate of $y$ is always 1.
 Each instruction (increment/decrement/zero check) in the two counter machine corresponds to a widget 
in the SHA. A widget simulating an instruction $l_1: c := c+1$;  goto  $l_2$ 
starts with a mode labeled $l_1$, and ends at a mode labeled $l'_2$. 
The mode labeled $l'_2$ is connected by an edge to a widget whose start mode 
is labeled $l_2$. This edge has guard $y=1$ and resets $y$. 
This forces a time elapse of atleast 1 each time 
this widget is taken. 

Consider the following widget which simulates an increment instruction
$l_1$: increment $c_1$, goto $l_2$. 
\begin{figure}[h]
\begin{center}
\begin{tikzpicture}[->,>=stealth',shorten >=1pt,auto,node distance=1.8cm,
  semithick]
  \tikzstyle{every state}=[fill=black!30!white,minimum size=3em,rounded rectangle]
  
  \node[initial,state] (A) {$\begin{array}{c}l_1 \\ -6 \end{array}$} ;
   \node[state] at (4,0) (B) {$\begin{array}{c}A \\ 3 \end{array}$} ;
  \node[state] at (4,-3) (C) {$\begin{array}{c}B \\ -3 \end{array}$} ;
  \node[state] at (0,-3) (D) {$\begin{array}{c}l'_2 \\ 0 \end{array}$} ;
    \node[state,fill=white,draw=none] at (-2,-3) (E) {} ;
    \path (A) edge node {$x{=}0?$} (B);
  \path (B) edge  node {$x=2?$} (C);
  \path (C) edge  node[above]{$y=1?$} node[below]{$y:=0$} (D);
\path (D) edge  node [above]{$y=1?$}
node [below]{$y:=0$} (E);
  \end{tikzpicture}
\caption{Simulation of increment instruction }
\label{inc-ctl}
\end{center}
\end{figure}
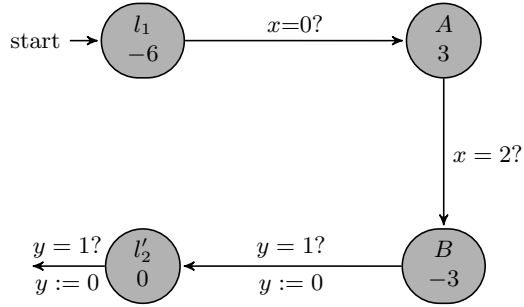
The CTL formula that takes us through the increment widget shown in figure \ref{inc-ctl} is  
$\varphi_{l_1,l_2}$ given by $(l_1 \Rightarrow \E \Diamond (l'_2 \wedge \E \bigcirc l_2))$.
To begin, $x=1$ (have an initialization widget). 

At the start of each increment/decrement/zero check widget, we have 
$y=0, x=2-\frac{1}{2^{c_1}3^{c_2}}$, where $c_1,c_2$ are the currrent counter values.
To reach $A$ from $l_1$ in figure \ref{inc-ctl},  a time $\frac{1}{3}- \frac{1}{
2^{c_1+1}3^{c_2+1}}$ is spent at $l_1$,
and hence $y=\frac{1}{3}- \frac{1}{2^{c_1+1}3^{c_2+1}}$. To reach $B$, a time of 
 $\frac{2}{3}$ is spent at $A$, obtaining $y=1- \frac{1}{2^{c_1+1}3^{c_2+1}}$. 
 Finally, $l'_2$ is reached with $y=0, x=2-\frac{1}{2^{c_1+1}3^{c_2}}$.
One time unit is spent at $l'_2$, and  we reach $l_2$ with the same $x,y$ values. 
The decrement instruction for $l_1$ can be obtained 
 by changing the rate of $x$ to -12 in mode $B$ in figure \ref{inc-ctl}; the CTL formula 
 for the decrement instruction remains same. 
 
 Now we consider the zero check instruction $l$: if $c_1=0$ goto $l_1$ else goto $l_2$.
 Figure \ref{zc} depicts the zero check widget. 
 
\begin{figure}[h]
\begin{center}
\begin{tikzpicture}[->,>=stealth',shorten >=1pt,auto,node distance=1.8cm,
  semithick]
  \tikzstyle{every state}=[fill=black!30!white,minimum size=3em,rounded rectangle]
   \node[initial,state] (A) {$\begin{array}{c}l \\ 1 \end{array}$} ;
   \node[state] at (2,2) (B) {$\begin{array}{c}A \\ 1 \end{array}$} ;
  \node[state] at (2,-2) (C) {$\begin{array}{c}B \\ 1 \end{array}$} ;
   \node[state,diamond] at (5,2) (B1) {$W1$} ;
   \node[state] at (2,0) (B2) {$\begin{array}{c}l'_1 \\ 0 \end{array}$} ;
    \node[state,diamond] at (5,-2) (C1) {$W2$} ;
   \node[state] at (5,0) (C2) {$\begin{array}{c}l'_2 \\ 0 \end{array}$} ;
    \path (A) edge node {$y{=}0?$} (B);
    \path (A) edge node {$y{=}0?$} (C);
   \path (B) edge  node {$y=0?$} (B1);
   \path (B) edge  node {$y=0?$} (B2);
  \path (C) edge  node {$y=0?$} (C1);
\path (C) edge  node {$y=0?$} (C2);
  \end{tikzpicture}
\caption{ZeroCheck}
\label{zc}
\end{center}
\end{figure}
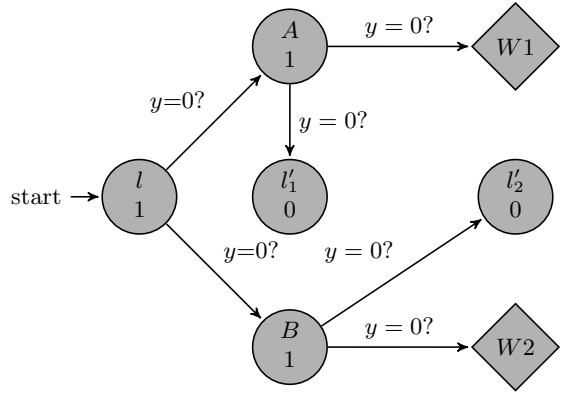
$W1$ and $W2$ are widgets that enforce that the zero check is done correctly. $W1$ checks if $c_1$ is indeed 0. 
If so, $x=2-\frac{1}{3^{c_2}}$. Likewise, $W2$ checks if $c_1 \neq 0$. If so, 
 $x=2-\frac{1}{2^{c_1}3^{c_2}}$, with $c_1 >0$. As earlier, $l'_1,l'_2$ are connected respectively to $l_1$ and $l_2$ 
 via transitions with guard $y=1$ and reset $y$.
The widgets $W1,W2$ are as follows:

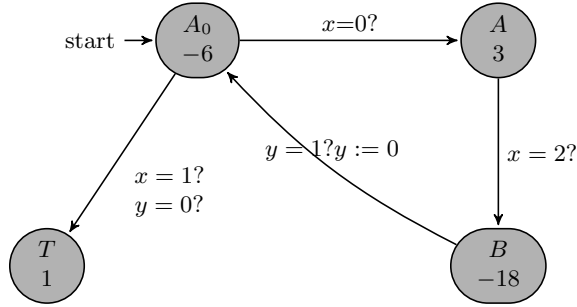
\begin{figure}[h]
\begin{center}
\begin{tikzpicture}[->,>=stealth',shorten >=1pt,auto,node distance=1.8cm,
  semithick]
  \tikzstyle{every state}=[fill=black!30!white,minimum size=3em,rounded rectangle]
  
  \node[initial,state] (A) {$\begin{array}{c}A_0 \\ -6 \end{array}$} ;
   \node[state] at (4,0) (B) {$\begin{array}{c}A \\ 3 \end{array}$} ;
  \node[state] at (4,-3) (C) {$\begin{array}{c}B \\ -18 \end{array}$} ;
  \node[state] at (-2,-3) (D) {$\begin{array}{c}T \\ 1 \end{array}$} ;
    \path (A) edge node {$x{=}0?$} (B);
  \path (B) edge  node {$x=2?$} (C);
  \path (C) edge[bend left=10]  node[above] {$y=1?y:=0$} (A);
\path (A) edge  node {$\begin{array}{c}x=1?\\y=0?\end{array}$} (D);
  \end{tikzpicture}
\caption{Widget $W1$}
\label{w1-ctl}
\end{center}
\end{figure}

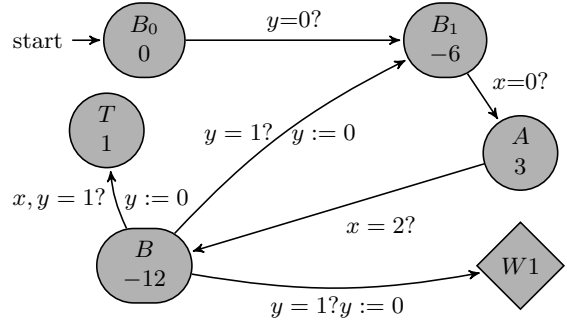
\begin{figure}[t]
\begin{center}
\begin{tikzpicture}[->,>=stealth',shorten >=1pt,auto,node distance=1.8cm,
  semithick]
  \tikzstyle{every state}=[fill=black!30!white,minimum size=3em,rounded rectangle]
  
  \node[initial,state] (A1) {$\begin{array}{c}B_0 \\ 0 \end{array}$} ;
  \node[state] at (4,0) (A) {$\begin{array}{c}B_1 \\ -6 \end{array}$} ;
   \node[state] at (5,-1.5) (B) {$\begin{array}{c}A \\ 3 \end{array}$} ;
  \node[state] at (0,-3) (C) {$\begin{array}{c}B \\ -12 \end{array}$} ;
  \node[state] at (-0.5,-1.2) (D) {$\begin{array}{c}T \\ 1 \end{array}$} ;
     \node[state,diamond] at (5,-3) (W) {$W1$} ;
    \path (A1) edge node {$y{=}0?$} (A);
    \path (A) edge node {$x{=}0?$} (B);
  \path (B) edge  node {$x=2?$} (C);
  \path (C) edge[bend left=10]  node[left] {$y=1?$}
                                node [right]{$y:=0$} (A);
   \path (C) edge[bend left=10]  node[left] {$x,y=1?$}
   node[right] {$y:=0$}(D);
  \path (C) edge[bend right=10]  node[below] {$y=1?y:=0$} (W);
  \end{tikzpicture}
\caption{Widget $W2$}
\label{w2-ctl}
\end{center}
\end{figure}

Widget $W1$ (figure \ref{w1-ctl}) repeatedly multiplies $\frac{1}{2^{c_1}3^{c_2}}$ 
by 3 till $x$ becomes 1. If this happens, then indeed $c_1=0$, and 
the mode $T$ is reached. Similarly, widget $W2$ first repeatedly mutiplies 
$\frac{1}{2^{c_1}3^{c_2}}$  by 2 (one or more times), till we obtain $\frac{1}{3^{c_2}}$, and then repeatedly multiplies 
by 3 (zero or more times) till $x=1$. 
Similarly, widget $W2$ (figure \ref{w2-ctl}) first repeatedly multiplies $\frac{1}{2^{c_1}3^{c_2}}$ 
by 2, and then (optionally) by 3 until $x$ becomes 1. 

The CTL formula to be checked at 
$l$ (in figure \ref{zc}) is $\varphi_{l,zc}$ given by 
$(l \Rightarrow [ (\E\bigcirc(A \wedge \E\bigcirc (l'_1 \wedge \E \bigcirc l_1) \wedge \E\Diamond T)) 
    \vee$\\
    $(\E\bigcirc(B \wedge \E\bigcirc (l'_2 \wedge \E \bigcirc l_2) \wedge
    \E\Diamond T))])$.
  
There is a mode labelled HALT, which corresponds to the halt instruction. Let $l_0$ 
be the label of the first instruction. 
Note that a minimum time of one unit is spent in all widgets. 
There is a self loop at the HALT mode with guard $y=1$, reset $y$. 
Any run is therefore, time divergent.
The final formula to be checked is $\Phi$ given  by 
\[
l_0 \wedge \E\{[\bigwedge (\varphi_{l_i,l_j} \wedge \varphi_{l,zc})]  \until{}
\mbox{ HALT }\}.
\] 

\subsubsection{Proof of (a) in Theorem~\ref{thm:dec-SHA-one-var}}

  The $\NLOGS$-hardness of the reachability problem for SHA follows from the
  complexity of reachability problem for finite graphs~\cite{jones1976new}.  
  On the other hand, the $\NLOGS$-hardness of the schedulability problem for SHA follows
  from the complexity of nonemptiness problem for B\"uchi automata (Proposition
  10.12 of ~\cite{PP04}).  
  For $\NLOGS$-membership of these problems we adapt the region construction for
  one-clock timed automata  proposed by Laroussinie, Markey, and
  Schnoebelen~\cite{LMS04}.  

  First notice that valuations of a one variable SHA can be written as real
  numbers corresponding to the value of the unique variable, hence we represent
  a configuration as $(m, x) \in M \times \Real$. 
  Also notice that the SHA compares valuations only against numbers appearing
  the constraints of the SHA.  
  Let $\bB = \set{b_0, b_1, \ldots, b_k}$ be the set of numbers appearing in the 
  definition of the SHA with the restriction that $b_0 = 0$ , and $b_i < b_j$
  for all $i < j$.    
  Observe that two configurations $(m, x)$ and $(m', x')$
  are indistinguishable by the constraints of an SHA if $m = m'$ and the
  values of the integer parts of $x$ and $x'$ are either both greater than $b_k$
  or both less than $b_k$, or the integer parts of $x$ and $x'$ are equal  and
  fractional parts are either both zero or both non-zero.  
  Given the largest constant appearing in the constraints of an SHA is $b_k$, then
  the regions (equivalence classes according to indistinguishablity by
  constraints) are the following the set of $4(b_k+1)$ intervals: 
  \begin{eqnarray*}
    (-\infty, -b_k), [-b_k, -b_k], (-b_k, -b_k+1), [-b_k+1,-b_k+1], \\
    \ldots [0,0], (0, 1), [1, 1], \ldots,  (b_k-1, b_k), [b_k, b_k], (b_k, \infty).
  \end{eqnarray*}
  Also note that such regions are also indistinguishable with respect to
  time-progress, as in any mode, the corresponding rate is  positive, negative,
  or zero, and in any of these cases, all the valuations in a given region  will
  pass through same set of regions as time progresses.  
  One can easily show that this region relation form the classical time-abstract
  bisimulation~\cite{BK08}, and hence the reachability and schedulability
  analysis of one variable SHA one can reduce~\cite{AD94,LMS04} these problems
  to corresponding problems on finite graphs whose vertices are the pair of
  modes and regions. 
  However, this construction does not yield an $\NLOGS$ algorithm, since these
  regions can not be stored in a logarithmic space algorithm.
  Using the similar observation from~\cite{LMS04}, it follows that the following set
  of regions also form time-abstract bisimulation over the configuration of SHA: 
  \begin{eqnarray*}
    (-\infty, -b_k), [-b_k, -b_k], (-b_{k}, -b_{k-1}), [-b_{k-1},-b_{k-1}], \\
    \ldots [b_0, b_0], (b_0, b_1), [b_1, b_1], \ldots,  (b_{k-1}, b_k), [b_k, b_k], (b_k, \infty).
  \end{eqnarray*}
  Moreover the number of regions in this encoding are linear in the size of
  the SHA, and hence a region can be stored in logarithmic space.
  We call this region graph LMS region graph.
  Using standard breadth-first search algorithms, the reachability and the
  schedulability problems for SHA can be solved in linear time. 
  Using this region graph, the standard $\NLOGS$ algorithms for reachability and
  repeated reachability over graphs can be adapted to give $\NLOGS$ algorithms
  for reachability and schedulability problems for one variable SHAs.  \hfill

\subsection{Proof of Lemma~\ref{thm:sha-one-var-nz-main}}
An infinite \emph{run} in the region graph $\Rr\Gg_{\Hh}$ of a one-variable SHA $\Hh$
is an infinite sequence of the form $((m_{0},r_{0}), a_{0}, (m_{1},r_{1}),
a_{1}, \ldots)$
where $r_{j}$ are valid regions and for all $j>0$ either $((m_{j},r_{j}),
a_{j},(m_{{j+1}},r_{{j+1}})) \in \Delta_{\Rr\Gg_{\Hh}}$ or $m_{j} = m_{{j+1}}$
and $a_{j} = \varepsilon$ and $r_{{j+1}}$ is a strict time-successor of
$r_{j}$ with respect to the mode $m_{j}$.
An infinite run in the region graph $\Rr\Gg_{\Hh}$ of a one-variable
SHA $\Hh$ is called \emph{progressive} iff it has a non-Zeno instantiation.
A region $r$ in $\Rr\Gg$ is called \emph{thick} if $\exists \delta>0, \nu \in r$ such that
$\nu + \delta \in r$. A region which is not \emph{thick} is called a \emph{thin} region.

 We first prove the ``only if'' side by proving that if none of the above conditions hold for a path, then the
only instantiations that the path has are Zeno.
Consider an infinite region path $((m_{0},r_{0}), a_{0}, (m_{1},r_{1}), a_{1}, \ldots)$
such that none of the above conditions hold. Then, one of the following three cases arise:
\begin{itemize}
\item There is an index $n$ such that for all $j\geq n$, we have $r_{j} = r$ where $r \models x>C_x$ and that $F(m_{j}) < 0$.
Consider the suffix starting from the $n^{th}$ index of any instantiation $\sigma$ , say 
$\sigma_n = ((m_{n},\nu_{n}), a_{n}, (m_{{n+1}},\nu_{{n+1}}), a_{{n+1}},
\ldots)$. 
If $\sigma$ was non-Zeno, then there must be a valuation $\nu$ in $\sigma_n$
such that $\nu < \nu_{n} - t*\rho_{min}$ for any $t\geq 0$, where $\rho_{min}
\geq 1$ is the minimum absolute value of any rate amongst the rates of $m_{n},
m_{{n+1}}, \ldots$. 
If we choose $t = \nu_{n} - C_x$, then $\nu < C_x$, violating our assumption.
\item There is an index $n$ such that for all $j\geq n$, we have $r_{j} = r$ where $r \models x<c_x$ and that $F(m_{j}) > 0$.
This case is similar to the previous one and hence omitted.
\item There is an index $n$ such that for all $j\geq n$
we have $r_{j} = r$ such that $r \models x \geq c_x$ and $r \models x \leq C_x$ and either $F(m_{j}) > 0$,
or $F(m_{j}) < 0$ for all $j>n$. W.l.o.g, we can assume that $F(m_{j}) > 0$ (The other case is analogous). Depending upon
the region $r$, we have the following two cases:
\begin{itemize}
	\item $r$ is of the form $x = c$, where c is an integer in $[c_x, C_x]$. Then all the instantitations
	of the path will clearly be Zeno, with the variable staying at the same value $c$ even though the rates in each of
	the locations visited are strictly positive.
	\item $r$ is of the form $c_1 < x < c_2$, where $c_1$ and $c_2$ are integers in $(c_x, C_x)$. Now, consider
	the suffix (starting from index $n$) of any instatiation $\sigma$ of the path, say $\sigma_n =$ 
	$((m_{n},\nu_{n}), a_{n}, (m_{{n+1}},\nu_{{n+1}}), a_{{n+1}}\ldots)$.
	Now, if $\sigma$ was non-Zeno, there must be a valuation $\nu$ in $\sigma_n$ such that $\nu > \nu_{n} + t*\rho_{min}$
	for any $t$, where $\rho_{min} \geq 1$ is the rate with the minimum absolute vale  amongst the rates of $q_{n}, q_{{n+1}}, \ldots$.
	But, then we can choose $t$ to be $c_2 - c_1 + 1$, in which case $\nu \not\in (c_1, c_2)$, giving a contradiction.
\end{itemize}
\end{itemize}
Hence, the conditions given are necessary for a path to be progressive.

We now prove that the above characterization is also sufficient. We do a case-by-case analysis for all the
5 conditions given:

\begin{itemize}
	\item {\bf Case-1} : Condition~\ref{nz-1} holds. In this case, consider any instantiation
	$\sigma = ((m_{0},\nu_{0}), a_{0}, (m_{1},\nu_{1}), a_{1}, \ldots)$. Then there are infinitely many indices
	${k_0}, {k_1}, \ldots$ such that the rates in the modes $m_{{k_0}}, m_{{k_1}}, \ldots$
	are zero. Then, it is indeed possible to take the discrete steps $a_{{k_0}}, a_{{k_1}} \ldots$
	after a unit interval of time, while still having the same abstract path. Such a run is non-Zeno.
	\item {\bf Case-2}: Condition~\ref{nz-2} holds. In any instantiation of the region path satisfying
	this condition, all valuations beyond a index $n$ satisfy $x>C_x$ and there are infinitely many indices
	$n < {k_0} < {k_1} < \cdots$ such that the rates in the modes $m_{{k_0}}, m_{{k_1}}, \ldots$ are strictly
	positive. Then, consider an instantiation in which the delay between the discrete steps $a_{{k_l} - 1}$ and $a_{{k_l}}$
	is equal to one. Clearly, this is a valid instantiation of the path, and is also non-Zeno.
	\item {\bf Case-3}: Condition~\ref{nz-3} holds. This case is similar to the previous case and hence, omitted.
	\item {\bf Case-4}: Condition~\ref{nz-4} holds. Note that, in this case, there exists 2 regions $r_1$ and $r_2$
	such that $r_1$ is a is a \emph{thin} region,  $r_2$ \emph{thick} region and
	there are infinitely many indices ${k_0}, {k_1}, \ldots$ such that $r_{{k_q}} = r_1$ and $r_{{k_q + 1}} = r_2$
	and $m_{{k_q}} = m_{{k_q + 1}}$. Hence, it is indeed possible to spend a 
	time of $\frac{1}{2*\rho}$ in the region $r_2$, where $\rho$ is the maximum absolute value of the rate
	in any mode, before taking the next discrete jump, while still respecting the abstract path. This gives a non-Zeno path.
	\item {\bf Case-5}: Condition~\ref{nz-5} holds. In any instantiation of the region path satisfying
	this condition, all valuations beyond a index $n$ satisfy $c_1 < x < c_2$ and there are infinitely many indices
	$n < {k_0} < {l_0} < {k_1} < {l_1} < \cdots$ such that the rates in the modes
	$m_{{k_0}}, m_{{k_1}}, \ldots$ are strictly positive, while those in the modes $m_{{l_0}}, m_{{l_1}}, \ldots$
	are strictly negative. W.l.o.g, we can choose $r_{n}$ to be one of $x = c_1$ or $x = c_2$. Also,
	if $r_{n} = (x=c_1)$, then, $F(m_{{n+1}}) > 0$ and we can assume ${k_0} = {n+1}$. For some fixed constant $\delta_1$,  one
	can then spend a time of $\frac{\delta_1}{F(m_{{k_0}})}$ in mode $m_{{k_0}}$ and $0$ units in all modes between ${k_0}$
	and ${l_0}$ (both exclusive). Also, we can then spend a time of $\frac{\delta_1 - \delta_2}{\vert F(m_{{l_0}}) \vert}$ in mode $m_{{l_0}}$,
	where $0 < \delta_2 < \delta_1 < 1$ are some fixed constants. Similarly, one can spend:
	\begin{itemize}
		\item $\frac{\delta_1 - \delta_2}{F(m_{{k_q}})}$ units of time in mode $m_{{k_q}}$ for $q > 0$
		\item $\frac{\delta_1 - \delta_2}{\vert F(m_{{l_q}})\vert }$ units of time in mode $m_{{l_q}}$ for $q > 0$
		\item $0$ units of time, otherwise
	\end{itemize}
	Note that, between indices ${k_q}$ and ${l_q}$, a minimum of $\frac{\delta_1 - \delta_2}{\rho_{max}}$ is spent,
	where $\rho_{max} \geq 1$ is the maximum absolute value of any rate amongst the rates of modes $m_{j}$, $j \geq k_1$.
	Hence, such a run is non-Zeno.
\end{itemize}

The above conditions can be formulated as an LTL formula and conjuncted with the
original LTL formula so that Zeno behaviours are ruled out.
We define predicates $F_0, F_{+}, F_{-}$ to denote rates 0, positive and negative, 
which hold good mode $m$ iff $F(m)=0, F(m)>0$ and $F(m)<0$ respectively. 
We also define $open_C$ and $open_c$ which hold good for regions of the form $(C_x, \infty)$ and 
$(-\infty,c_x)$ respectively. Further, we have predicates $thick$ and $thin$ representing \emph{thick} and \emph{thin} regions.
\textsc{Regions} is the set of all regions.

Here are the LTL formulae for conditions 1-5:
\begin{enumerate}
\item $\Box \Diamond F_0$
\item $\Diamond\{\Box (open_C \wedge \Diamond F_+)\}$
\item Similar to 2.
\item $\bigvee\limits_{r \neq r' \in \textsc{Regions}}[\Box \Diamond r \wedge \Box \Diamond r']$
\item $\bigvee\limits_{r \in \textsc{Regions}}\Diamond\{\Box[thick \wedge r \wedge \Diamond F_+ \wedge \Diamond F_-]\}$
\end{enumerate}

Note that the conditions specified in Lemma~\ref{thm:sha-one-var-nz-main} are a superset
of the conditions for a progressive run in a timed automata ~\cite{AD94}.
This is because, in an SHA, the variables can change with negative rates.
Consider for example, the following system with $x$ as a continuous variable:

\begin{figure}
  \begin{center}
    \begin{tikzpicture}
      \tikzstyle{every state}=[fill=gray!20!white,minimum size=2em,shape=rounded rectangle]

      \node[initial,state,fill=gray!20] (m1) {$m_0$};
      \node[state,fill=gray!20] at (4, 0) (m2) {$m_1$};
      
     \path[->] (m1) edge[bend left] node [above] {$x<1$}  (m2);
     \path[->] (m2) edge[bend left] node [below] {$x<1$}  (m1);

    \end{tikzpicture}
  \end{center}
\vspace{-1em}
\caption{Progressive paths}
\vspace{-2em}
\label{fig:progressive}
\end{figure}
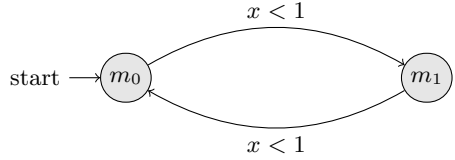

If the automaton in Figure~\ref{fig:progressive} is a timed automaton, then there is no
non-Zeno run starting with $x=0$. If however, the system was an SHA with rates
$F(m_0) = +1$ and $F(m_1) = -1$, then the following is a valid non-Zeno run:
$((m_0, 0),(m_1, 0.5),(m_0, 0.25),(m_1, 0.5),(m_0, 0.25),(m_1, 0.5),\ldots)$,
illustrating why the conditions in ~\cite{AD94} are not complete for progressive runs
in systems with negative rates.

\subsubsection{Proof of Theorem~\ref{thm:wsha-reach-np-hard}: Reachability}

  This is a continuation of the proof given in the main paper.
  To show NP-hardness of the reachability we reduce the \emph{subset-sum problem} to solving the
  reachability problem in a WSHA. 
  Given the set $A$ and the integer $k$, we construct a WSHA $\Hh$ with $n+3$
  variables $x_0,x_1, \dots, x_{n+2}$, $2n+1$ modes $m_0, \dots, m_{2n}$ and
  $2n$ transitions, such that a particular target polytope $\Tt$ is reachable in
  the WSHA iff there is a non-empty subset $T \subseteq A$ that sums up to $k$.
  The rates in the modes $m_i (0\leq i\leq 2n)$ are given as follows ($\vr_i$ represents
  
  Given the set $A$ and the integer $k$, we construct a WSHA $\Hh$ with $n+3$
  variables $x_0,x_1, \dots, x_{n+2}$, $2n+1$ modes $m_0, \dots, m_{2n}$ and
  $2n$ transitions, such that a particular target polytope $\Tt$ is reachable in
  the WSHA iff there is a non-empty subset $T \subseteq A$ that sums up to $k$.

The rates in the modes $m_i (0\leq i\leq 2n)$ are given as follows ($r_i$ represents
$(F(m_i))$:
\begin{itemize}
\item[$\bullet$] $\vr_0(x_0) =1, \vr_0(x_{n+1}) = \vr_0(x_{n+2}) = 0$, and
$\vr_0(x_i) = a_i$, where $1\leq i\leq n$
\item[$\bullet$] $\vr_{2j-1}(x_j) = -a_j$, $\vr_{2j-1}(x_{n+1}) = a_j$,
$\vr_{2j-1}(x_{n+2}) = 1$, and $\vr_{2j-1}(x_0)$ = $\vr_{2j-1}(x_i) = 0$,
where $1\leq i\neq j\leq n$
\item[$\bullet$] $\vr_{2j}(x_j) = -a_j$, and $\vr_{2j}(x_0) = \vr_{2j}(x_i) =$
$\vr_{2j}(x_{n+1}) = \vr_{2j}(x_{n+2}) = 0$,
 where $1\leq i\neq j\leq n$
\end{itemize}
The transitions in the WSHA are as follows:
\begin{itemize}
\item[$\bullet$] There are edges from $m_0$ to $m_1$ and to $m_2$.
\item[$\bullet$] There are edges from $m_{2j-1}$ and from $m_2j$ to $m_{2j+1}$ and to $m_{2j+2}$, $1\leq j\leq n-1$
\end{itemize}
We now claim that the polytope $\Tt$, given by the set of points $\px$
such that $\px(0) = 1$, $\px(i) = 0$, $1\leq i\leq n$, $\px(n+1) = k$
and $\px(n+2) \in [1,n]$, is reachable from the point $\vzero$
iff there is a subset $T$ of $A$ that sums up to $k$.

Before proving the correctness of the construction, we present a brief intuition of what
each variable stands for. The variable $x_0$ ensures that the variables
$x_1, x_2\ldots x_n$ are initialized with values $a_1, a_2, \ldots a_n$ (the elements
of $A$). The variable $x_{n+1}$ sums up the values of the elements in the chosen
subset $T$, and can be later compared with $k$.
The variable $x_{n+2}$ ensures that the set $T$ is non-empty (specifically when 
$k = 0$).

We first prove that if there is a non-empty $T$ that sums to $k$, then $\Tt$
is reachable from $\vzero$.
It is evident that the point $(1,0,0\cdots,k,|T|) \in \Tt$ can be reached
by taking the following sequence of (time, mode) having length $n+1$:
\[
\seq{(1,m_0), (1, l_1), (1, l_2)\cdots, (1, l_n)}
\]
where $l_i = m_{2i-1}$ if $a_i \in T$, otherwise $l_i = m_{2i}$.

We now prove that if the WSHA $\Hh$ reaches a point in $\Tt$ starting from $\vzero$,
then there is a non-empty subset $T$ of $A$ that sums to $k$.
For this, note that since $m_{2j-1}$ or $m_{2j}$ ($1\leq j\leq n$) cannot be
reached from one another in the underlying graph, atmost one of them can
occur in any path in the WSHA.
Also, since $\pt(0) = 1$, $\forall \pt \in \Tt$, and since for all modes $m$ apart from
$m_0$, $F(m)(x_0) = 0$, the mode $m_0$ must be taken for exactly unit time.
This means that the first step in any run that reaches $\Tt$ is $(t_0 = 1, m_0)$.
Now, after the first step, the WSHA reaches the valuation $(1, a_1, a_2\cdots a_n,0,0)$.
In order to reach $\Tt$, all the coordinates $x_1$ to $x_n$ must go to zero.
However, since the rate of $a_i$ is non-zero only in the modes $m_{2i-1}$ and $m_{2i}$, and
because atmost one of $m_{2i-1}$ and $m_{2i}$ can be taken, we conclude that exactly one
of $m_{2i-1}$ and $m_{2i}$ is taken for a 1 unit of time.
(For the sake of simplicity, we can assume wlog that each of $a_i\neq 0$). 
 Hence, any reachable run will comprise of $n+1$ steps.
Also, since we demand that $\pt(n+2) \in [1,n]$ and the rate of $x_{n+2}$ is non-zero
only in the modes $m_{2i-1}$ ($1\leq i\leq n$), atleast one of the modes $m_{2i-1}$
is necessarily taken. 
Now consider the run that reaches $\Tt$:
\[
\seq{(1, m_0), (1, l_1), (1, l_2)\cdots, (1, l_n)}
\]
It is clear that the set $\{a_i | l_i = m_{2i-1}\}$ is a non-empty subset that sums to $k$.
Figure \ref{nphard-main} gives an illustration of the WSHA construction for a set $\{1,2,-3\}$. 

\subsubsection{Proof of Theorem~\ref{thm:wsha-reach-np-hard}: Schedulability}
For schedulability, observe that given a run type $\sigma = \seq{n_0, b_1, n_1, \ldots, b_p, n_p}$ and an initial
valuation $\nu_0$, we say that $\sigma$ schedules the WSHA $\Hh$ from $\nu$ to $\Tt$ if the
following linear program is feasible: 
for every $ 0 \leq i \leq p$ and  $m \in M_{n_i}$ there are $\nu_{n_i},
\nu_{n_i}' \in \Real^{|X|}$ and  $t_i^m \in \Rplus$ such that     
\begin{eqnarray}
  \nu_0 &=& \nu_{n_0} \nonumber \\
  \nu_{n_i}, \nu_{n_i}'  & \in & S_{M_{n_i}} \text{ for all $0 \leq i \leq p$}
  \nonumber \\
  \nu_{n_i}   & \in & G(b_i) \text{ for all $0 < i \leq p$}
  \nonumber \\
  \nu_{n_{i+1}}(j) &=& 0  \text{ for all $x_j \in R(b_{i+1})$ and  $0 < i \leq p$}
  \nonumber \\
  \nu_{n_{i+1}}(j) &=& \nu_{n_i}'(j)  \text{ for all $x_j \not \in R(b_{i+1})$ and  $0 < i \leq p$}
  \nonumber \\
  \nu_{n_i}'  & = & \nu_{n_i} + \sum_{m \in M_{n_i}} F(m) \cdot t_i^m \text{
    for all $0 \leq i \leq p$} \nonumber\\
  t_i^m  &\geq&  0 \text{ for all $0 \leq i \leq p$ and $m \in M_{n_i}$} \nonumber\\ 
  \vzero &=& \sum_{m \in M_{n_p}} F(m) \cdot t_p^m  \nonumber\\
  1 &=& \sum_{m \in M_{n_p}} t_p^m \nonumber
\end{eqnarray}
These constraints check whether it is possible to reach some mode with the
highest rank while satisfying the guard and constraints of the WSHA, and
schedule the system forever in the CMS defined by the SCCs of the highest rank in
the type. The proof for this claim is similar to the proof for
the CMS, and hence omitted.
Since the size of a run type is polynomial in the size of WSHA, and
checking whether a run type schedules from $\nu$ can be performed in
polynomial time (LP feasibility), the schedulability problem in WSHA can be decided in
polynomial time.  
The NP-hardness of the result follows from a similar reduction from subset sum
problem as shown in Theorem~\ref{thm:wsha-reach-np-hard}. 

\subsubsection{Proof of Corollary~\ref{thm:wsha-ltl-mc}}

  Let $(\Hh, \Pp, L)$ be a weak Kripke SHA and let $\phi$ be an LTL property. 
  From the standard LTL to B\"uchi automata construction~\cite{VW86} the negation
  of $\phi$ can be converted to B\"uchi automata $\Aa_{\neg \phi}$  whose size
  is exponential in the size of $\phi$.
  The LTL model checking problem then reduces to solving a schedulability
  problem for the product of original automaton and $\Aa_{\neg \phi}$. 
  We observe that the product of a weak SHA $\Hh$ and $\Aa_{\neg \phi}$ remains
  weak SHA, since variables occur only in the WSHA.  
  The LTL model checking problem thus reduces to the schedulability problem for
  a weak SHA of size exponential in the size of original problem, and hence can
  be decided using standard polynomial space algorithm.
  PSPACE-hardness of the problem follows as the LTL model checking over finite
  automata is already PSPACE-complete.

\subsection{\underline{Proofs from Section~\ref{sec:undec}}}

\subsubsection{Proof of Theorem~\ref{thm:sha-undec-three-var}}

  We show a reduction from the halting problem for two-counter Minsky machines
  ${\cal M}$. 
  The variables $x,y,z$ of the SHA  have global invariants $0 \leq x \leq 1$, $0
  \leq y,z\leq 5$ respectively. 
  The counters $c_1, c_2$ of two-counter machine are encoded in variables $y$ and
  $z$ as 
  \[
  y=5-\frac{1}{2^{c_1}} \text{ and  } z=5-\frac{1}{2^{c_2}}.
 \]
  To begin, we have $c_1=c_2=0$, hence $y=z=4$, and $x=0$.
  The rates of $x,y,z$ are indicated by 3-tuples inside the modes of the SHA. 
  Likewise, the discrete updates on $x,y,z$ are also indicated by a 3-tuple 
  on the transitions.  We construct widgets for each of the increment/decrement and zero check instructions. 
  Each widget begins with $x=0,y=5-\frac{1}{2^{c_1}}$ and $z=5-\frac{1}{2^{c_2}}$, where $c_1,c_2$ are the current 
  values of the counters. 
  \begin{itemize}
  \item {\bf (Increment and Decrement Instructions).}
    Let us first consider increment instruction 
    \[
    l : c_1 := c_1 + 1 \text{ goto } l'.
    \] 
    The Figure~\ref{inc-c1} depicts the increment widget. 
    This widget starts with a mode labeled $l$, and  ends in a mode labeled $l'$. 
    On entering $l$, we have $x=0$ and $y=5-\frac{1}{2^{c_1}}$. 
    Observe that a time equal to $k=\frac{1}{6}[5-\frac{1}{2^{c_1}}]$ units has 
    to be spent at $l$, giving $x=\frac{1}{6}[5-\frac{1}{2^{c_1}}]$ and $y=0$. 
    Otherwise, if one spends a time larger than $k$, say $k+\eta$, at $l$, then we obtain
    $y=5-\frac{1}{2^{c_1}}-6[k+\eta] <0$, violating the invariants of
    $y$. 
    Similarly, if a time $k-\eta$ is spent at $l$,  
    then we obatin $y=5-\frac{1}{2^{c_1}}-6[k-\eta]=6\eta$, which would make $y=5+6 \eta > 5$ on taking the transition from 
    $l$ to $A$. Thus, on entering $A$, we have $y=5$ and $x=k$. At $A$, a time $k_1=\frac{1}{6}$ is spent, obtaining 
    $x=k+k_1=1-\frac{1}{6.2^{c_1}}$, and $y=0$. If a time $<k_1$ is spent at $A$, we get $y=5-30(k_1-\kappa)=30 \kappa$
    which would violate the invariants  of $y$ on taking the transition from $A$ to $B$. Similarly, 
    if a time $>k_1$ was spent at $A$, we get $y=5-30(k_1+\kappa)<0$ again violating the invariants of $y$. 
    Thus, on reaching $B$, we have $x=1-\frac{1}{6.2^{c_1}}$ and $y=5$. To reach $l'$, we spend 
    $\frac{1}{6.2^{c_1}}$ time at $B$, and get $x=0$, and $y=5-3(\frac{1}{6.2^{c_1}})=5-\frac{1}{2^{c_1+1}}$. 
    Note that if a time less than $\frac{1}{6.2^{c_1}}$ is spent at $B$, then $x$ would violate 
    its invariants while taking the transition to $l'$.  Likewise, 
    if a time $>\frac{1}{6.2^{c_1}}$ is spent at $B$, $x$ would violate its
    invariants.
    
    This widget can be modified to simulate the instructions increment counter $c_2$,
    decrement counters $c_1$ and decrement counter $c_2$,respectively,  by changing  the cost rate of $z$ at
    $B$ to -3, the cost rate of $y$ at $B$ to -12, and the cost rate of $z$ at $B$
    to $-12$, respectively.  
    
  \item {\bf (Zero Check Instruction).}
    We next consider the zero check instruction:
    \[
    l: \text{ if } c_1=0 \text{ goto } l' \text{ else  goto } l''.
    \]
 The right side widget of Figure~\ref{z-c1} depicts the zero check widget. 
 The zero check widget  starts in a mode $l$ and reaches either the mode $l'$ or mode $l''$.

 \begin{figure}
  \begin{center}
  \scalebox{0.7}{
  \begin{tikzpicture}[->,>=stealth',shorten >=1pt,auto,node distance=1.8cm,
    semithick]
    \tikzstyle{every state}=[fill=black!10!white,minimum size=3em,rounded rectangle]
    \node[initial,state, initial text={}] at (-4,0) (A1) {$\begin{array}{c}l \\ (1,-6,0)\end{array}$} ;
    \node[state] at (0,0) (A) {$\begin{array}{c}A \\ (1,-30,0)\end{array}$} ;
    \node[state] at (4,0) (B) {$\begin{array}{c}B \\ (1,-3,0)\end{array}$} ;
    \node[state] at (8,0) (C) {$\begin{array}{c}l' \\ (x,y,z)\end{array}$} ;
    \path (A1) edge node [above]{(0,5,0)}(A);
    \path (A) edge node [above]{(0,5,0)}(B);
    \path (B) edge node [above]{(-1,0,0)}(C); 
  \end{tikzpicture}
  }
  \caption{Increment $c_1$ widget}
   \label{inc-c1}
  \end{center}
\end{figure}
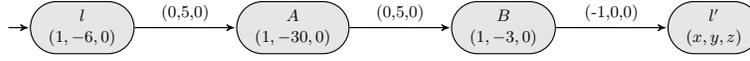

\begin{figure}
  \begin{center}
  \scalebox{0.7}{
  \begin{tikzpicture}[->,>=stealth',shorten >=1pt,auto,node distance=1.8cm,
    semithick]
    \tikzstyle{every state}=[fill=black!10!white,minimum size=3em,rounded rectangle]
    \node[initial,state,initial where = above, initial text={}] at (3,0) (A1) {$\begin{array}{c}l \\ (2,1,0)\end{array}$} ;
    \node[state] at (6,0) (A) {$\begin{array}{c}A \\ (-1,-1,0)\end{array}$} ;
    \node[state] at (9,0) (A2) {$\begin{array}{c}l' \\ (x,y,z)\end{array}$} ;
    \node[state] at (0,0) (E) {$\begin{array}{c}B \\(1,\frac{1}{2},0) \end{array}$} ;
    \node[state] at (-3,0) (B) {$\begin{array}{c}C \\(2,1,0) \end{array}$} ;
    \node[state] at (-6,0) (B2) {$\begin{array}{c}D \\ (0,0,0)\end{array}$} ;
    \node[state] at (-9,0) (C2) {$\begin{array}{c}l'' \\ (x,y,z)\end{array}$} ;
    \path (A1) edge node [above]{(0,1,0)}(A);
    \path (A) edge node [above]{(0,-1,0)}(A2);
    \path (B) edge node [above]{(-1,0,0)}(B2);
    \path (B2) edge node [above]{(0,4,0)}(C2);
    \path (A1) edge node [above]{(0,-5,0)}(E);
    \path (E) edge node [above]{(-1,0,0)}(B);
  \end{tikzpicture}
}
\caption{Zero Check widget}
\label{z-c1}
\end{center}
\end{figure}
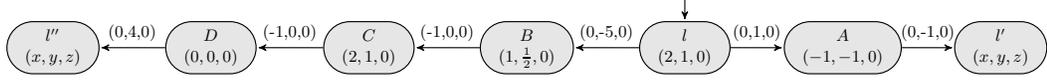
   
 Assume that $c_1=0$. Then $y=4$. Then the only possible transition to take is the one from $l$ to $A$, which 
makes $y=5$. Note that no time is spent at $l$ in this case, since it would violate 
the invariants of $y$ while moving to $A$. Thus, at $A$ we have $y=5,x=0$. No time is spent at $A$ since 
it would make $x <0$. On reaching $l'$, we therefore have $y=4$. Note that if $c_1=0$, then we cannot take the transition 
from $l$ to $B$ since it cannot meet invariants $y \leq 5$
and $x \leq 1$ simultaneously : Unless one time unit is spent at $l$, $y$ would violate its invariants 
on taking the edge to $B$. In this case, $x$ will become 2, violating $0 \leq x \leq 1$.

If $y=5-\frac{1}{2^{c_1}}$ with $c_1 >0$, then we take the transition from $l$ to $B$ after spending some time 
at $l$. We cannot take the transition from $l$ to $A$, since that would violate 
the invariants of $y$. Thus, we go to $B$ when $c_1 >0$. The time spent at $l$ is some $0<\eta < 1$. 
This gives $y=5-\frac{1}{2^{c_1}}+\eta$ and $x=2\eta$. 
If $\eta > \frac{1}{2^{c_1}}$, then the invariants of $y$ get violated. If $\eta < \frac{1}{2^{c_1}}$, then 
on taking the transition to $B$, we obtain $y=5-\frac{1}{2^{c_1}}+\eta-5<0$, again violating 
the invariants of $y$. 
Hence, $\eta=\frac{1}{2^{c_1}}$, and we reach $B$ with $x=2\eta,y=0$.
Now at $B$, a time $1-2\eta$ is spent, reaching $C$ with $x=0, y=\frac{1}{2}-\eta$. 
If a time $< 1-2\eta$ is spent at $B$, then $x$ will violate its invariants 
on taking the transition to $C$; if a time $>1-2\eta$ is spent at $B$, $x$ will exceed 1.
At $C$, we thus have $x=0$ and $y=\frac{1}{2}-\frac{1}{2^{c_1}}$. 
At $C$, $\frac{1}{2}$ units of time is spent, obtaining $x=1$ and $y=1-\frac{1}{2^{c_1}}$, 
so that on reaching $D$, $x=0$ and $y=1-\frac{1}{2^{c_1}}$. Finally, $l''$ is reached with 
$x=0$ and $y=5-\frac{1}{2^{c_1}}$. 
\end{itemize}
It is straightforward to see that the modules for increment, decrement, and zero
check simulate the two counter machine. There is a mode HALT corresponding to the HALT  instruction. 
The halting problem for two counter machines is thus  reduced to the
reachability of the mode HALT.

\subsubsection{Proof of Theorem~\ref{thm:sha-undec-three-var-one-clock}}
 We simulate a two counter machine using a CMS with 3 variables and one clock. 
 The three variables $x_1,x_2,y$ have global invariants $0 \leq x_1, x_2 \leq 5$ and $0 \leq y \leq 1$ 
 respectively. The clock variable is $x$. 
 The counters $c_1,c_2$ are encoded as $x_1=5-\frac{1}{2^{c_1}}$, $x_2=5-\frac{1}{2^{c_2}}$. 
 At the beginning of each widget, we have $x_1=5-\frac{1}{2^{c_1}}$, $x_2=5-\frac{1}{2^{c_2}}$ and $y=1$, where 
 $c_1, c_2$ are the current counter values. 
 To begin with, $x_1=x_2=4$, $y=1$ and $x=0$. 
  The rates of $x_1,x_2,y$  are written inside the modes of the CMS as 3 tuples. 
 
 We first see the simulation of an increment instruction $l$: increment $c_1$, goto $l'$. 
 Figure \ref{inc-2} depicts the increment widget. This widget starts with a mode $l$ 
 and goes to a mode $l'$.

  \begin{figure}[h]
\begin{center}
 \scalebox{0.7}{
\begin{tikzpicture}[->,>=stealth',shorten >=1pt,auto,node distance=1.8cm,
  semithick]
  \tikzstyle{every state}=[fill=black!10!white,minimum size=1em,rounded rectangle]
 \node[initial,state, initial text={}] at (-2.5,0) (A1) {$\begin{array}{c}l \\ (6,0,-1)\end{array}$} ;
  \node[state] at (2.5,0) (A) {$\begin{array}{c}A \\ (-5,0,0)\end{array}$} ;
   \node[state] at (2.5,-2) (B) {$\begin{array}{c}B \\ (5,0,0)\end{array}$} ;
   \node[state] at (-2.5,-2) (C) {$\begin{array}{c}C \\ (-3,0,1)\end{array}$} ;
   \node[state] at (-2.5,-4) (D) {$\begin{array}{c}D \\ (0,0,-1)\end{array}$} ;
    \node[state] at (2.5,-4) (E) {$\begin{array}{c}E \\ (0,0,1)\end{array}$} ;
      \node[state] at (2.5,-6) (F) {$\begin{array}{c}l' \\ (x_1,x_2,y)\end{array}$} ;
    \path (A1) edge node [above]{$0<x<1?$}
		    node [below]{$x:=0$}(A);
    \path (A) edge node [left]{$x=1?$}
                    node [right]{$x:=0$}(B);
\path (B) edge node [above]{$x=1?$}
                    node [below]{$x:=0$}(C);
\path (C) edge  node [left]{$x:=0$}(D);
\path (D) edge node [above]{$x=1?$}
                    node [below]{$x:=0$}(E);
                    
   \path (E) edge node [left]{$x=1?$}
                    node [right]{$x:=0$}(F);
                 
  \end{tikzpicture}
  }
\caption{Increment $c_1$}
\label{inc-2}
\end{center}
\end{figure}
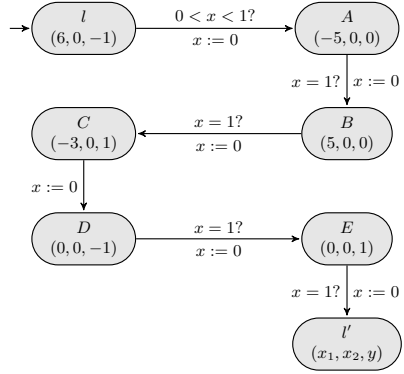

We start at mode $l$ with $x_1=5-\frac{1}{2^{c_1}}$ and $y=1$. A time 
$k=\frac{1}{6.2^{c_1}}$ is spent at $l$ obtaining $x_1=5$ and $y=1-k$. 
Note that spending time $>k$ at $l$ would make $x_1 >5$, while 
spending time $<k$ would give $y=1-k+6\eta$ and $x_1=5-6\eta$ for some $\eta < k$. 
In this case, one time unit spent at $A$ gives $x_1<0$, violating 
the invariant. Hence, $k$ time units are spent at $l$, and at $B$, 
we obtain $x_1=0$ and $y=1-k$. We reach $C$ with $x_1=5$ and $y=1-k$.
At $C$, a time of $k$ is spent, giving $x_1=5-\frac{1}{2^{c_1+1}}$, 
and $y=1$. Note that spending time $>k$ at $C$ will give $y >1$. 
If a time $<k$ is spent at $C$, then 
$y <1$ on reaching $D$, and this results in $y <0$ at $E
$, violating the invariant for $y$. Thus, 
we indeed spend $k$ time units at $C$, and reach $E$ with 
$y=0$ and $x_1=5-\frac{1}{2^{c_1+1}}$. We then reach $l'$ with 
$y=1$ and $x_1=5-\frac{1}{2^{c_1+1}}$.

The decrement $c_1$ operation can be done by changing the rate 
-3 of $x_1$ at $C$ to $-12$.

Next, we will see the zero check instruction $l$: if $c_1=0$ goto $l'$ 
else goto $l''$. Figure \ref{z-c2} depicts the zero check widget.

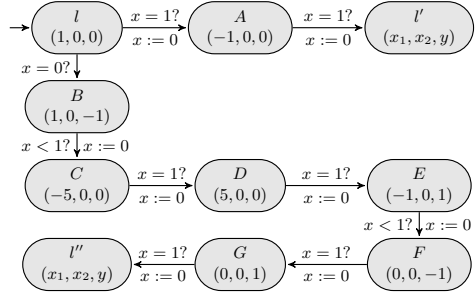
\begin{figure}[h]
\begin{center}
 \scalebox{0.7}{
\begin{tikzpicture}[->,>=stealth',shorten >=1pt,auto,node distance=1.8cm,
  semithick]
\tikzstyle{every state}=[fill=black!10!white,minimum size=3em,rounded rectangle]
 \node[initial,state,initial text={}] at (-1.5,0) (A1) {$\begin{array}{c}l \\ (1,0,0)\end{array}$} ;
  \node[state] at (1.6,0) (A) {$\begin{array}{c}A \\ (-1,0,0)\end{array}$} ;
   \node[state] at (5,0) (A2) {$\begin{array}{c}l' \\ (x_1,x_2,y)\end{array}$} ;
  \node[state] at (-1.5,-1.5) (E) {$\begin{array}{c}B \\(1,0,-1) \end{array}$} ;
    \node[state] at (-1.5,-3) (B) {$\begin{array}{c}C \\(-5,0,0) \end{array}$} ;
   \node[state] at (1.6,-3) (B2) {$\begin{array}{c}D \\ (5,0,0)\end{array}$} ;
   \node[state] at (5,-3) (C2) {$\begin{array}{c}E \\ (-1,0,1)\end{array}$} ;
    \node[state] at (5,-4.5) (F) {$\begin{array}{c}F \\ (0,0,-1)\end{array}$} ;
     \node[state] at (1.6,-4.5) (G) {$\begin{array}{c}G \\ (0,0,1)\end{array}$} ;
     \node[state] at (-1.5,-4.5) (H) {$\begin{array}{c}l''\\ (x_1,x_2,y)\end{array}$} ;
    \path (A1) edge node [above]{$x=1?$}
                         node [below]{$x:=0$}(A);
\path (A) edge node [above]{$x=1?$}
                         node [below]{$x:=0$}(A2);
   \path (A1) edge node [left]{$x=0?$}(E);
\path (E) edge node [left]{$x<1?$}
                         node [right]{$x:=0$}(B);
\path (B) edge node [above]{$x=1?$}
                         node [below]{$x:=0$}(B2);
\path (B2) edge node [above]{$x=1?$}
                         node [below]{$x:=0$}(C2);
\path (C2) edge node [left]{$x<1?$}
                         node [right]{$x:=0$}(F);
\path (F) edge node [above]{$x=1?$}
                         node [below]{$x:=0$}(G);
\path (G) edge node [above]{$x=1?$}
                         node [below]{$x:=0$}(H);
\end{tikzpicture}
}
\caption{Zero Check}
\label{z-c2}
\end{center}
\end{figure}
We start at mode $l$ with $y=1$ and $x_1=5-\frac{1}{2^{c_1}}$. Assume $c_1=0$. In this case, 
we cannot take the transition to $B$, since 
$x_1$ would attain a negative value on entering $D$. Thus, the only possibility is to reach $l'$ 
with values of $x_1,y$ unchanged. 

Now assume $c_1>0$. Then it is not possible to take the transition to $A$
since that would result in $x_1 >5$. Thus, we move to $B$, with values unchanged. 
At $B$, a time $\eta <1$ is elapsed, resulting in $x_1=5-\frac{1}{2^{c_1}}+\eta$ and $y=1-\eta$.
Clearly, $\eta \leq \frac{1}{2^{c_1}}$. If $\eta < \frac{1}{2^{c_1}}$, then on reaching $D$, 
we will obtain $x_1 <0$, thus $\eta=\frac{1}{2^{c_1}}$. Then, at $D$, we have 
$x_1=0$ and $y=1-\frac{1}{2^{c_1}}$. Next, $E$ is reached with $x_1=5$ and $y=1-\frac{1}{2^{c_1}}$.
At $E$, we spend a time $\kappa <1$, obtaining $x_1=5-\kappa$ and $y=1-\frac{1}{2^{c_1}}+\kappa$. 
If $\kappa\neq \frac{1}{2^{c_1}}$, then on reaching $G$, $y$ will violate its invariants; hence
at $G$ we have $x_1=5-\frac{1}{2^{c_1}}$ and $y=0$. We then reach $l''$ with $y=1$ and 
$x_1=5-\frac{1}{2^{c_1}}$.

It can be seen that the widgets for increment/decrement/zero check faithfully simulate 
the two counter machine. There is a mode HALT corresponding to 
the HALT instruction. Thus, the HALT mode can be reached iff the two counter machine halts. 


\end{document}